\documentclass[sigconf]{acmart}

\AtBeginDocument{%
  }

\setcopyright{none}
\usepackage{tikz}
\usepackage{amsmath}
\usepackage{amsfonts}
\usepackage{amsthm}
\usepackage{algorithm}
\usepackage{algorithmic}
\usepackage{dsfont}
\usepackage{enumitem}
\usepackage{graphicx}
\usepackage{subcaption}
\usepackage{xspace}
\usepackage{bm}
\usepackage{booktabs}
\usepackage{multirow}
\usepackage[title,titletoc]{appendix}
\usepackage{filecontents}
\usepackage{ulem}

\newcommand{\mybpara}[1]{\noindent{\textbf{#1}}\xspace}

\newcommand{\MSE}{\ensuremath{\mathsf{MSE}}\xspace}

\newcommand{\IGR}{\ensuremath{\mathsf{IGR}}\xspace}

\newcommand{\Diffstats}{\ensuremath{\mathsf{Diffstats}}\xspace}
\newcommand{\ASD}{\ensuremath{\mathsf{ASD}}\xspace}
\newcommand{\RSN}{\ensuremath{\mathsf{RSN}}\xspace}

\newcommand{\blackcirclea}[1]{%
    \begin{tikzpicture}[baseline=(text.base)]
        \node[
            fill=black, 
            text=white, 
            shape=circle, 
            inner sep=0.5pt, 
            minimum size=10pt,
            align=center,
        ] (text) {\raisebox{1.4pt}{\textbf{#1}}};
    \end{tikzpicture}%
}

\newcommand{\blackcircleb}[1]{%
    \begin{tikzpicture}[baseline=(text.base)]
        \node[
            fill=black, 
            text=white, 
            shape=circle, 
            inner sep=0.5pt, 
            minimum size=10pt,   % Fixed size for the circle
            font=\fontsize{8}{8}\selectfont,
            align=center,
        ] (text) {\raisebox{1pt}{\textbf{#1}}};
    \end{tikzpicture}%
}
\newtheorem{definition}{Definition}
\newcounter{alphasubsection}

\newcommand\blfootnote[1]{%
  \begingroup
  \renewcommand\thefootnote{}\footnote{#1}%
  \addtocounter{footnote}{-1}%
  \endgroup
}
\renewcommand\footnotetextcopyrightpermission[1]{} % removes footnote with conference information in first column
\begin{document}
%-------------------------------------------------------------------------------

%don't want date printed
\date{}

% make title bold and 14 pt font (Latex default is non-bold, 16 pt)
% \title{\Large \bf Restoring Utility of Local Differential Privacy under Data Poisoning Attack}
% \title{\Large \bf Mitigating Data Poisoning Attacks to Local Differential Privacy}
\title{Mitigating Data Poisoning Attacks to Local Differential Privacy}
 \author{Xiaolin Li}
 \affiliation{%
   \institution{Purdue University}
   \city{West Lafayette}
   \state{IN}
   \country{USA}
 }
 \email{li4955@purdue.edu}

  \author{Ninghui Li}
 \affiliation{%
   \institution{Purdue University}
   \city{West Lafayette}
   \state{IN}
   \country{USA}
 }
 \email{ninghui@purdue.edu}
 
 \author{Boyang Wang}
 \affiliation{%
   \institution{The University of Cincinnati}
   \city{Cincinnati}
   \state{OH}
   \country{USA}
 }
 \email{boyang.wang@uc.edu}

  \author{Wenhai Sun}
 \affiliation{%
   \institution{Purdue University}
   \city{West Lafayette}
   \state{IN}
   \country{USA}
 }
 \email{whsun@purdue.edu}
% \author{Lars Th{\o}rv{\"a}ld}
% \affiliation{%
%   \institution{The Th{\o}rv{\"a}ld Group}
%   \city{Hekla}
%   \country{Iceland}}
% \email{larst@affiliation.org}

% \author{Valerie B\'eranger}
% \affiliation{%
%   \institution{Inria Paris-Rocquencourt}
%   \city{Rocquencourt}
%   \country{France}
% }

% \author{Aparna Patel}
% \affiliation{%
%  \institution{Rajiv Gandhi University}
%  \city{Doimukh}
%  \state{Arunachal Pradesh}
%  \country{India}}

% \author{Huifen Chan}
% \affiliation{%
%   \institution{Tsinghua University}
%   \city{Haidian Qu}
%   \state{Beijing Shi}
%   \country{China}}

% \author{Charles Palmer}
% \affiliation{%
%   \institution{Palmer Research Laboratories}
%   \city{San Antonio}
%   \state{Texas}
%   \country{USA}}
% \email{cpalmer@prl.com}

% \author{John Smith}
% \affiliation{%
%   \institution{The Th{\o}rv{\"a}ld Group}
%   \city{Hekla}
%   \country{Iceland}}
% \email{jsmith@affiliation.org}

% \author{Julius P. Kumquat}
% \affiliation{%
%   \institution{The Kumquat Consortium}
%   \city{New York}
%   \country{USA}}
% \email{jpkumquat@consortium.net}

% \renewcommand{\shortauthors}{Trovato et al.}

\begin{abstract}
The distributed nature of local differential privacy (LDP) invites data poisoning attacks and poses unforeseen threats to the underlying LDP-supported applications. In this paper, we propose a comprehensive mitigation framework for popular frequency estimation, which contains a suite of novel defenses, including malicious user detection, attack pattern recognition, and damaged utility recovery. 
In addition to existing attacks, we explore new adaptive adversarial activities for our mitigation design. For detection, we present a new method to precisely identify bogus reports and thus LDP aggregation can be performed over the ``clean'' data. When the attack behavior becomes stealthy and direct filtering out malicious users is difficult, we further propose a detection that can effectively recognize hidden adversarial patterns, thus facilitating the decision-making of service providers. These detection methods require no additional data and attack information and incur minimal computational cost. Our experiment demonstrates their excellent performance and substantial improvement over previous work in various settings. In addition, we conduct an empirical analysis of LDP post-processing for corrupted data recovery and propose a new post-processing method, through which we reveal new insights into protocol recommendations in practice and key design principles for future research. %Our code is available at \url{https://anonymous.4open.science/r/MDPA_LDP-F362/}. 
\end{abstract}

%\begin{CCSXML}
%<ccs2012>
%   <concept>
%       <concept_id>10002978.10002991.10002995</concept_id>
%       <concept_desc>Security and privacy~Privacy-preserving protocols</concept_desc>
%       <concept_significance>500</concept_significance>
%       </concept>
% </ccs2012>
%\end{CCSXML}

%\ccsdesc[500]{Security and privacy~Privacy-preserving protocols}

%%
%% Keywords. The author(s) should pick words that accurately describe
%% the work being presented. Separate the keywords with commas.
%\keywords{Local Differential Privacy, Data Poisoning Attack, Attack Detection, Post-Processing, Data Recovery}
%% A "teaser" image appears between the author and affiliation
%% information and the body of the document, and typically spans the
%% page.

% \received{20 February 2007}
% \received[revised]{12 March 2009}
% \received[accepted]{5 June 2009}

\maketitle
\pagestyle{plain} 
%-------------------------------------------------------------------------------

%-------------------------------------------------------------------------------
% \input{1}

\section{Introduction}
Local differential privacy is a promising privacy-enhancing tool~\cite{duchi2013local,wang2017locally,li2020estimating,duchi2018minimax,wang2019collecting, ye2019privkv} and widely used in many applications ~\cite{erlingsson2014rappor,team2017learning,wang2019answering,qin2016heavy}. Recent studies show that LDP is vulnerable to data poisoning attacks, i.e., its results can be manipulated by a small portion of malicious local users~\cite{cao2021data,cheu2021manipulation, wu2022poisoning, li2023fine,tong2024data, li2024robustness}.  This emerging threat urges people to rethink the security implications of LDP and highlights the pressing need for effective defense for vulnerable LDP protocols. 
\blfootnote{This paper was accepted by ACM CCS 2025 (https://doi.org/10.1145/3719027.3744839).}

\vspace{-5pt}
% related works mitigation limitation: ineffective detection with intuitive attacks, little known about the recoverability of post processing.
The existing research centered on attack exploration and understanding of adversarial behavior, while little attention has been paid to a systematic study of countermeasures. In this work, we focus on mitigating data poisoning attacks on state-of-the-art categorical frequency oracles (CFOs) for frequency estimation \cite{wang2017locally, cheu2021manipulation, kairouz2014extremal}. In the literature, the detection methods were briefly discussed supplementary to the main effort of attack discovery \cite{cao2021data} and suffer from many limitations. For instance, in the state-of-the-art maximal gain attack (MGA) \cite{cao2021data}, the attackers set the target item indexes to 1 in their reports to boost the corresponding frequencies. Meanwhile, they also set a fixed number of non-target indexes to masquerade as benign users. Despite being intuitive, such a fixed number may leak attack patterns. Therefore, a practical detection should consider stealthier attacks that can adaptively set non-target indexes (see Section~\ref{Attack_intro}). In addition, the existing detection methods often require extra knowledge and incur non-negligible false positives and time costs.  %such as intuitive attack strategies, unrealistic extra knowledge, non-negligible accuracy loss and computational overhead. \xl{For instance, the intuitive maximal gain attack (MGA) sets the target index to 1 to boost the estimated frequency of the target item, while randomly setting some untargeted indices to 1 so that the total number of 1 bits matches the expectation for benign users (see MGA detail in Section~\ref{Attack_intro}). However, an attacker can adaptively modify the reported vector’s 1-bit settings to further obscure this fixed pattern (see APA in Section~\ref{Attack_intro})}. 
On the other hand, post-processing as a utility boosting method for LDP was also adapted to tackle post-attack data recovery \cite{sun2024ldprecover, huang2024ldpguard, cao2021data}. However, there is a lack of investigations on their comparative performance and understanding of the relationship between the for-attack methods and those for no-attack settings, which is crucial for attack-aware post-processing design. In this paper, we want to answer the following research questions to provide guidance in developing a reliable and informed mitigation framework that is adaptive to various challenging attack scenarios. 

% RQ: 1. attack sophistication. 2. fake user detection; 3. corrupted data detection; 4 post processing
% \xl{The template specifies that Modifying the template --- including but not limited to: adjusting margins, typeface sizes, line spacing, paragraph and list definitions, and the use of the \texttt{\textbackslash verb|\textbackslash vspace|} command to manually adjust the vertical spacing between elements of your work --- is not allowed.}

\vspace{4pt}
\textbf{RQ-1.} \textit{Can the attacker go beyond the state-of-the-art MGA~\cite{cao2021data} and launch a more adaptive attack as mentioned above for elevated stealthiness?}  MGA is an intuitive adversarial strategy. A discussion on more advanced attack behaviors is paramount for us to understand the attacker's capabilities and build a meaningful foundation for better defense. %\sun{only emphasize the question}

\vspace{4pt}
\textbf{RQ-2.} \textit{Can we accurately and efficiently identify malicious users?} The current fake user detection targets MGA only with high false positives \cite{cao2021data}. It also incurs significant computational overhead, which is not friendly to time-sensitive applications. 

\vspace{4pt}
\textbf{RQ-3.} \textit{Can we still effectively detect attack behaviors without additional data and attack knowledge when identifying bogus reports is challenging?} Fake user detection is not always possible against adaptive, stealthy attacks. It is important to capture abnormal behavior with minimal cost for rapid and informed decision-making. 

\vspace{4pt}
\textbf{RQ-4.} \textit{Given positive results from the detections mentioned in \textbf{RQ-2} and \textbf{RQ-3}, can we recover the corrupted data utility? How good are the current post-processing methods for attack suppression and utility boosting}? Prior work studied the attack \cite{cao2021data,sun2024ldprecover} and non-attack \cite{wang2019locally} settings independently. Discussing them with comparative analyses helps us understand their respective enabling components and interpret their performance for practical recommendations and new designs in the future.

\vspace{4pt}
In this paper, we study defenses against advanced data poisoning attacks on state-of-the-art CFOs protocols, i.e., GRR~\cite{kairouz2014extremal}, OUE~\cite{wang2017locally}, OLH~\cite{wang2017locally} and HST \cite{cheu2021manipulation}. We investigate extending existing attacks to new LDP settings and study the potentials to achieve more hidden attack behaviors. We propose a new metric to effectively measure the attack efficacy across various LDP protocols and data settings. 

%\vspace{2pt}
Based on a realistic set of attack variants, we propose a mitigation framework consisting of novel attack detections and attack-resilient post-processing methods for data recovery. In particular, we present a new \textit{differential statistical anomaly detection} to find fake LDP participants controlled by the attacker.  The method leverages the abnormal patterns of bogus reports to identify malicious users. %It is sensitive to the number of malicious users in the system while remaining accurate for adaptive attack behavior. 
Our experiment shows a substantial performance improvement compared to the prior work \cite{cao2021data} (e.g., about 0.8 versus 0.3 for $F1$ scores in most cases for MGA) while enabling fake user identification in an adaptive MGA setting for the first time. 
To respond to the challenges of detecting bogus reports with our new attack strategy, we propose an \textit{abnormal statistics detection} by using the inherent LDP characteristics to statistically differentiate the attack and non-attack scenarios. The experiment reports a detection accuracy of $100\%$ in almost all tested settings. We consider our detection methods \textit{zero-shot} since they only depend on known data knowledge and are agnostic to attack details. In addition, we significantly reduce the time costs for running detections and make them suitable for applications that are sensitive to latency in practice.

In this work, we further study utility recovery under the data poisoning attack. We propose a novel LDP post-processing method, \textit{robust segment normalization}, and empirically study the performance of state-of-the-art for attack and non-attack purposes at the same time. We find that the widely adopted post-processing method -- \textit{consistency} (i.e., non-negative frequency estimates and sum to 1)~\cite{wang2019locally}, plays a more important role in effectively recovering corrupted data than other existing strategies. Based on the experimental results, we make recommendations to help service providers select appropriate LDP post-processing methods under varying attack influences. The relevant code is provided at \textbf{\url{https://github.com/Marvin-huoshan/MDPA_LDP/}}. %\sun{need github link}
We summarize our contributions below. 
\begin{itemize}[leftmargin=*]
\item We comprehensively investigate the mitigation against the data poisoning attack on LDP frequency estimation, including new attack exploration, novel detection methods, and attack-resilient utility recovery, which advances the much-needed defensive technology development and generates new knowledge for LDP security. 
\vspace{2pt}
\item Our research reveals a new stealthy attack strategy on existing CFO protocols, deepening our understanding of attacker capabilities. We also present a metric to help measure the attack efficacy over diverse protocols and data settings. 
\vspace{2pt}
\item We propose novel zero-shot detection methods for underlying attacks to identify malicious users and hidden attack patterns. The experimental results show significant improvement in detection accuracy and time cost compared to the state-of-the-art. %thus friendly to service providers. 
\vspace{2pt}
\item We design a new LDP post-processing method and empirically study the impact of post-processing in the presence of attackers with varying capabilities for the first time. The study produces new insight into the recoverability of corrupted LDP data. The experimental results also help with recommendations for practical deployment. 
\end{itemize}

%-------------------------------------------------------------------------------
\section{Background and Related Work} 
\subsection{Local Differential Privacy} 
LDP enables $n$ users to share their data $v$ with an untrusted server through a local perturbation function $\Psi(\cdot)$ such that only obfuscated items $\Psi(v)$ is obtained by the server. Formally, 
\begin{definition} ($\epsilon\mbox{-Local Differential Privacy}$~\cite{duchi2013local}). An algorithm $\Psi(\cdot):\mathcal{D}\rightarrow \hat{\mathcal{D}}$ satisfies $\epsilon$-LDP if for any $v_1, v_2 \in \mathcal{D}$ and for $y \in \hat{\mathcal{D}}$, $\Pr[\Psi(v_1) = y] \leq e^\epsilon \Pr[\Psi(v_2) = y]$.\end{definition}

\subsubsection{Categorical Frequency Oracles}\label{CFOs}

We briefly introduce some state-of-the-art CFOs used in this paper for frequency estimation in LDP. We assume that there are $d$ items in $\mathcal{D}$. $v\in[d]$ is the index of the item in the encoding space,  where $[d]$ is $\{1,...,d\}$ for simplicity.

\vspace{2pt}
\textbf{Generalized Randomized Response (GRR)}~\cite{kairouz2014extremal}.
In GRR, a user directly encodes their item $v$. The perturbation function keeps $v$ with probability $p = \frac{e^\epsilon}{e^\epsilon + d - 1}$ and changes it to another item with probability $q = \frac{1}{e^\epsilon + d - 1}$. By collecting perturbed user report $y^{(j)}$ from all $n$ users, the frequency of each item can be estimated by the aggregation function $\Phi_{\mathrm{GRR}}(v) = \frac{C_v - nq}{n(p-q)}$, where $C_v$ is the count of $v$ instances in report $y$.

\vspace{2pt}
\textbf{Optimal Unary Encoding (OUE)} ~\cite{wang2017locally}.
OUE achieves the theoretical lower bound of $L_2$ errors. In OUE, each item $v$ is encoded into a one-hot vector $\mathbf{v} = [0, \ldots, 1, \ldots, 0]$ with 1 in the $v$-th position only. $\mathbf{v}$ is further randomized to the report $\mathbf{y}$ by flipping 1 to 0 with probability $p=\frac{1}{2}$ and 0 to 1 with probability $q=\frac{1}{e^\epsilon + 1}$. We obtain the frequency of item $v$  by $\Phi_{\mathrm{OUE}}(v) = \frac{\sum_{j=1}^{n} \mathbf{y}^{(j)}[v] - \frac{n}{e^\epsilon + 1}}{n(\frac{1}{2}-\frac{1}{e^\epsilon + 1})}$.

\vspace{2pt}
\textbf{Optimal Local Hashing (OLH)}~\cite{wang2017locally}. OLH can also achieve the same minimum $L_2$ error as OUE. It is preferred for large domain sizes by encoding input items to a smaller domain of size $g = \lfloor e^{\epsilon}+1\rfloor \ll d$. In the \textit{user setting} of OLH, each user randomly selects a hash function $h$ from a universal hash family $\mathbf{H}$ to encode $v \in [d]$ into $v_h \in [g]$. The perturbation function keeps $\hat{v}_h = v_h$ with probability $p = \frac{e^\epsilon}{e^\epsilon + g - 1}$ and changes it to another hash value with probability $q = \frac{1}{e^\epsilon + g - 1}$. Given the reports $y^{(j)}=\langle h^{(j)}, \hat{v}^{(j)}_h \rangle$ from all $n$ users, the frequency of each item can be estimated by  $\Phi_{\text{OLH}}(v) = \frac{C_v - \frac{n}{g}}{n(\frac{e^\epsilon}{e^\epsilon+g-1} - \frac{1}{g})}$, where $C_v = |\{j \mid h^{(j)}(v) = \hat{v}^{(j)}_h\}|$ counts the number of reports that support the item $v$.

\vspace{2pt}
\textbf{ExplicitHist (HST)}~\cite{cheu2021manipulation}. In HST, a uniform public vector $\mathbf{s}$ of length $d$ is generated. HST becomes OLH with $g=2$  \cite{wang2017locally}. The user $j$ randomizes the $v$-th element $\mathbf{s}^{(j)}[v]$ to $\frac{e^\epsilon + 1}{e^\epsilon - 1} \times \mathbf{s}^{(j)}[v]$ with probability $\frac{e^\epsilon}{e^\epsilon + 1}$ and to $-\frac{e^\epsilon + 1}{e^\epsilon - 1} \times \mathbf{s}^{(j)}[v]$ with probability $\frac{1}{e^\epsilon + 1}$. The frequency of each item thus can be estimated by $\Phi_{\text{HST}}(v) = \frac{1}{n} \sum_{j=1}^{n} y^{(j)} \times \mathbf{s}^{(j)}[v]$.

\subsubsection{LDP Post Processing}\label{existing_post_methods}
Consistency condition (i.e., non-negative frequency estimates and sum to 1) is widely used along with CFOs as a post-processing strategy to improve the utility~\cite{wang2019locally}. We briefly describe Norm-Sub and Base-Cut as they are based on consistency and do not assume prior data knowledge (versus Power and PowerNS~\cite{wang2019locally}). They were also recommended in prior work for non-attack situations~\cite{wang2019locally}. 
We denote the estimated frequency as $\tilde{f}$ and the one after post-processing as $\tilde{f}^{\prime}$.

\vspace{2pt}
\textbf{Norm-Sub.} Norm-Sub adjusts a frequency estimate to $\tilde{f}_v^{'}=\max(\tilde{f}_v + \Delta, 0)$ by converting negative estimates to zero and adding $\Delta$ to the remaining estimates. As a result, the sum of all frequency estimates equals 1, i.e., $\sum_{v \in \mathcal{D}} \max(\tilde{f}_v + \Delta, 0) = 1$. 

\vspace{2pt}
\textbf{Base-Cut}.
Base-Cut does not consider \textit{sum to 1}. Instead, it simply keeps the estimates that are above a sensitivity threshold and sets the rest to zero. 

Post-processing methods are also adopted for utility recovery under data poisoning attacks. 

\vspace{2pt}
\textbf{Normalization}. The server re-calibrates the frequency of each item $v$ by $\tilde{f}_v^{\prime} = \frac{\tilde{f}_v - \tilde{f}_{min}}{\sum_v(\tilde{f}_v - \tilde{f}_{min})}$, where $\tilde{f}_{min}$ is the smallest estimate~\cite{cao2021data}.%estimated frequency~\cite{cao2021data}.

\vspace{1pt}
\textbf{LDPRecover}. In LDPRecover~\cite{sun2024ldprecover}, the observed frequency $\tilde{f}_Z(v)$ of an item is assumed to be a combination of the genuine component $\tilde{f}_X(v)$ and bogus component $\tilde{f}_Y(v)$, weighted by the proportions of genuine and fake users.  LDPRecover divides items into two sets based on their observed frequencies and focuses on restoring genuine frequencies $\tilde{f}_X^{'}$ by solving a constraint inference problem. LDPRecover also ensures the consistency condition.

In addition, LDPGuard~\cite{huang2024ldpguard} used two rounds of reports to estimate attack details, including the percentage of fake users and attack strategies, to inform the recovery, which may not be practical. We do not consider it in this paper.  %\sun{put this part to the experiment}

 We propose a new post-processing method in this work and empirically study it with the above state-of-the-art approaches in both attack and non-attack environments. We expect this will generate new knowledge about the recoverability of corrupted utility and shed light on attack-resilient designs in the future.

% \vspace{-2pt}
\subsection{Attack Detection}
Frequent itemset anomaly detection (FIAD) was proposed in ~\cite{cao2021data} to detect fake users for MGA. The intuition is that the reports of fake users always support a set of target items regardless of LDP perturbation. The method adopted frequent itemset mining to find malicious users by checking their supported items.  However, FIAD fails when $r$ (i.e., the number of target items) is small and becomes ineffective against the adaptive MGA attack. FIAD also suffers from a high false positive rate (refer to ~\cite{cao2021data} for more details). A conditional probability-based attack detection was also presented in \cite{cao2021data} to detect polluted item frequencies instead of bogus users. It depends on ground-truth data knowledge, such as fake user percentage and attacked items, which may not be available in practice.

In this paper, we propose a novel method for fake user detection, which requires no prior data and attack knowledge. It significantly outperforms the state-of-the-art and produces fewer false positives. In addition, we present a new detection approach for scenarios where identifying bogus reports is challenging.
\section{Mitigation Overview}
In this section, we overview the proposed mitigation. We introduce the threat model, target attacks, and metrics for evaluation.

\subsection{Threat Model}
\vspace{2pt}
\mybpara{Attacker's Capability and Goal.} Consistent with prior work~\cite{cao2021data,li2023fine}, we assume that the attacker can control $m = \beta \cdot n$ fake users, where $n$ is the total number of users and $\beta \in [0, 1]$ is the percentage of fake users. Since the LDP perturbation function is on the user end, the attacker can circumvent it and directly craft bogus values in its output domain $\mathcal{\hat{D}}$. As a result, the fake reports will be injected and aggregated with benign ones on the server side. The attacker also knows the related information, such as privacy budget $\epsilon$, the item domain $\mathcal{D}$ and its size $d$, and the support set $S(y)$, which is a set of items that report $y$ supports \cite{wang2017locally}.

The attacker's goal is to increase the estimated frequencies of a set of  $r$ target items $T = \{t_1, t_2, ... , t_r\}$. To this end, the attacker carefully crafts the perturbed values $\hat{Y}$ to maximize the overall gain of target items: 
\begin{equation} \label{Attacker_goal}
   \max_{\hat{Y}}\sum_{t\in T}\mathbb{E}\left[\Delta\Tilde{f}_t\right]
\end{equation}
where $\Delta\Tilde{f}_t = \Tilde{f}_{t,\textit{after}} - \Tilde{f}_{t,\textit{before}}$ is the frequency gain of target item $t$ after the attack. We also consider a baseline attack with the same goal, but the attacker can only provide false values in the input domain of the perturbation. Thus its behavior is indistinguishable from honest users.

\vspace{2pt}
\mybpara{Knowledge of Defender}. We assume that the defender knows nothing about attack details (e.g., $\beta$ and $T$) and underlying data, other than the LDP parameters (e.g., $\epsilon$ and $n$) and received reports. 

\subsection{Attacks}\label{Attack_intro} 

\vspace{2pt}
\textbf{Maximal Gain Attack.} 
MGA is the state-of-the-art data poisoning attack on CFOs~\cite{cao2021data}. The main idea is to craft the perturbed values for the fake users via solving the optimization problem in Eq.~\eqref{Attacker_goal}. The original MGA only supports OUE, OLH-User (i.e., the user setting) and GRR. We extend it to OLH-Server (i.e., the server setting) and HST. We briefly describe it here and refer readers to \cite{cao2021data} for details of the original MGA. 
\begin{itemize}[leftmargin=*]
    \item \textbf{OUE.} For each fake user, the attacker initializes a zero vector $y^{(j)}$ of length $d$ and sets $y^{(j)}_t = 1$ for all $t\in T$. To further hide traces, the attacker randomly sets $l=\lfloor p + (d-1)q -r\rfloor$ non-target bits to ensure the number of 1's matches the expected number in the report of a genuine user. 
   
    \item \textbf{OLH-User.} In OLH-User, users randomly select the hash function. Therefore, each fake user will choose a hash function $h^{(j)}$ that maps all items in $T$ to $v^{(j)}_h$, i.e., $\sum_{t\in T} \mathds{1}_{S(y^{(j)})}(t)=r$ and submit the report $y^{(j)}=\langle h^{(j)}, v^{(j)}_h \rangle$ to the server.  $\mathds{1}_{S(y^{(j)})}(v)$ is a characteristic function and outputs 1 if $y^{(j)}$ supports item $t$. The number of items supported by the hash value is made close to $\sum_{v\in \mathcal{D}} \mathds{1}_{S(y^{(j)})}(v) = \frac{d}{g}$ to further hide the attack.

    \item \textbf{OLH-Server.} In OLH-Server, the server chooses the hash function $h^{(j)}$ such that the attacker aims to find a hash value $v_{h} = \arg\max\left(\sum_{t\in T} \mathds{1}_{S(y^{(j)})}(t)\right)$ that maps the most items $t \in T$ with $h^{(j)}$.
   
    \item \textbf{HST-User.} In this setting, a user samples her public vector $\mathbf{s}^{(j)}$. For each fake user, the attacker initializes a public vector $\mathbf{s}^{(j)} = [-1, -1, \ldots, -1]$ of length $d$, and sets $\mathbf{s}^{(j)}[t] = 1$ for all $t \in T$. Since $\mathbf{s}^{(j)}$ is expected to follow uniform distribution (i.e., equal numbers of $-1$ and $1$), the attacker randomly sets $l = \lfloor d/2 - r \rfloor$ non-target positions to $1$ in $\mathbf{s}^{(j)}$. $y^{(j)}$ is then set to $\frac{e^{\epsilon}+1}{e^{\epsilon}-1}$ to maximize the frequencies of the target set.
    
    \item \textbf{HST-Server.} In HST-Server, the server sets a public binary vector for each user uniformly. Fake users can only promote the expected frequencies of the aggregated target set, i.e., $\sum_{t\in T}\Big[\frac{1}{n}\sum_{j=1}^n$ $y^{(j)} \times \mathbf{s}^{(j)}[t]\Big]$, by manipulating $y^{(j)}$. 

    \item \textbf{GRR.} A user report is a perturbed item in GRR. We have $\sum_{t\in T}$ $ \mathds{1}_{S(y^{(j)})}(t)\leq 1$ and $\sum_{t\in T} \mathds{1}_{S(y^{(j)})}(t) = 1$ when $y^{(j)}$ is a target item in $T$.  Therefore, MGA selects any $t\in T$ for each fake user. 
\end{itemize}

\vspace{2pt}
\textbf{Adaptive Maximal Gain Attack (MGA-A).}
MGA-A is an enhanced MGA that can evade the FIAD detection~\cite{cao2021data} by crafting a perturbed value that only supports a subset of the target set $T$. Specifically, the attacker randomly selects and supports an itemset of size  $r'$ from $\binom{r}{r'}$ possible subsets for each fake user, where $r' < r$. MGA-A was shown to significantly reduce the effectiveness of FIAD, especially when $r'\le 2$, while maintaining high attack efficacy (refer to \cite{cao2021data} for more details). MGA-A was originally designed for OUE. In this paper, we extend it to cover OLH and HST.  MGA-A cannot be applied to GRR which only reports a single value.
% MGA-A is a stealthier version of the MGA under the OUE protocol. In the original MGA, the simultaneous occurrence of a large number of identical target items in the perturbed values can be easily detected by frequent item mining-based methods. To avoid detection, the adaptive version reduces the number of simultaneous occurrences of target items, making the attack less conspicuous.

% In MGA-A, each fake user crafts their perturbed value by supporting only a subset of the target items. Instead of including all $r$ target items, the attacker randomly selects $r'$ items to support, where $r' < r$. Specifically, the attacker divides the set of target items into smaller itemsets of size $r'$, resulting in $\binom{r}{r'}$ possible subsets. This reduces the likelihood of detection while maintaining the overall attack objective.

\vspace{2pt}
\textbf{Adaptive Pattern Attack.}\label{APDA}
Though MGA-A may evade FIAD, our detection method can still effectively identify fake users (see Section~\ref{sec:Result_Diffstats}). In this paper, we discover a new attack strategy, \textit{adaptive pattern attack} (APA), where the attacker can further hide attack patterns by strategically setting bits in the crafted report. In MGA-A, the support of the values in a crafted report is matched with the expected support of a genuine user, for instance, setting $\sum_{v\in \mathcal{D}} \mathds{1}_{S(y^{(j)})}(v) = \lfloor p + (d - 1)q \rfloor$ in OUE;  the proposed APA further tweaks the support at non-target positions to maximize the estimated frequencies of target items while reducing detectability. Specifically, APA can be applied to the following CFOs. \begin{itemize}[leftmargin=*]
    \item \textbf{OUE.} For the number $k \in [0, d]$ of 1's in the report produced by the LDP, APA generates $\omega[k]$, which represents the number of fake users whose perturbed reports support $k$ items, such that $\sum_{k}\omega[k] = m$. $m$ is the total number of fake users. For different $k$, $\omega[k]$ fake users are selected and $l = \max(\lfloor k-r \rfloor,0)$ non-target items are randomly set to 1 in their reports.
    
    \item \textbf{OLH-User.} For any user report $y^{(j)}$, let $k \in [0, d]$ be the number of supported items $v \in \mathcal{D}$. For all $k$, $\omega[k]$ is generated such that $\sum_{k} \omega[k] = m$. For a specific fake user, assuming its assigned target support count is $k$, APA  ensures that the items supported by the hash value $v^{(j)}_h$ are close to $k$, i.e., $\sum_{v\in \mathcal{D}} \mathds{1}_{S(y^{(j)})}(v) = k$.
    
    \item \textbf{HST-User.} For $k \in [0, d]$ bits set to 1 in a report,  $\omega[k]$  is generated and  $\sum_{k}\omega[k] = m$. For different $k$,  $\omega[k]$ fake users are selected and  $l = \max(\lfloor k - r \rfloor,0)$ non-target positions are randomly set to $1$ in the report. The report $y_j$ is then set to $\frac{e^{\epsilon}+1}{e^{\epsilon}-1}$ to maximize the frequencies of the target set.
\end{itemize}
Note that in the server setting of OLH and HST, the attacker is significantly constrained in selecting hash functions or setting bits in the public vector since these are assigned by the server; in GRR, the report is a single value. Therefore, we do not consider APA in these protocols.

\vspace{2pt}
\textbf{Baseline Attack.} The baseline attack is a universal strategy that follows an LDP protocol but supplies bogus data in the input domain of LDP \cite{cao2021data,li2023fine}. Therefore, the behavior of a baseline attacker is indistinguishable from that of an honest user without prior knowledge of the ground truth data. We use the baseline attack as a benchmark and propose a new metric based on it for a more meaningful interpretation and comparison between our recovery methods and existing ones (see Section~\ref{sec:metric}). 

\begin{figure*}[ht] 
\centering
\includegraphics[scale=0.31]{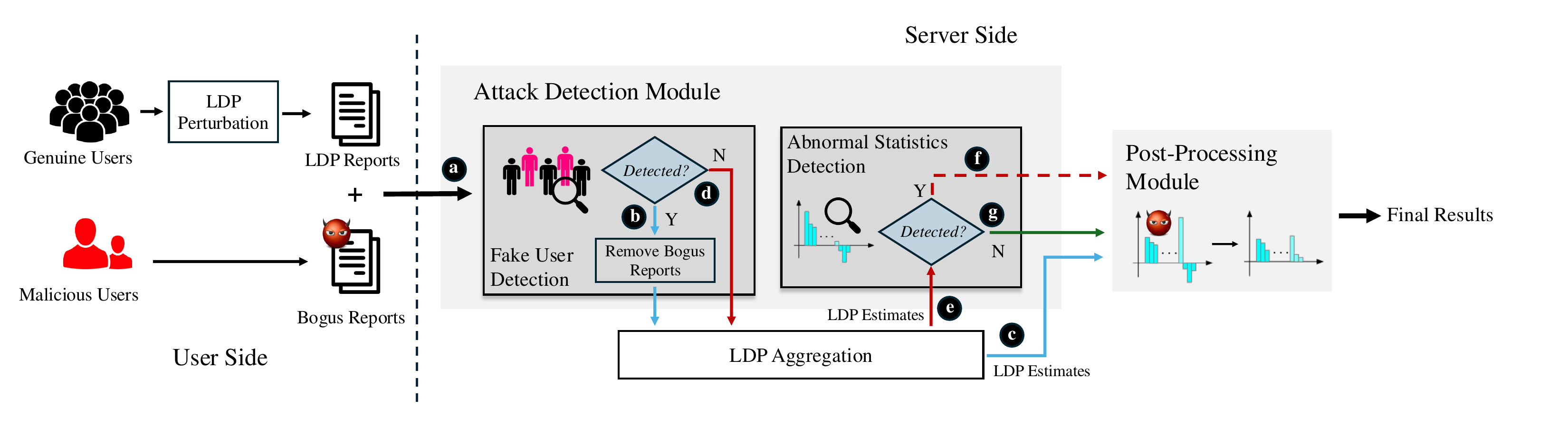}
% \vspace{-8pt}
\caption{Mitigation workflow for data poisoning attacks on LDP}\label{Framework}
\Description{Mitigation workflow for data poisoning attacks on LDP}
% \vspace{-10pt}
\end{figure*}
\subsection{Workflow}
We overview our mitigation for data poisoning attacks on LDP in Figure~\ref{Framework}. The defense occurs on the server side and contains two modules: \textit{attack detection} and \textit{post-processing}. We propose two novel detection methods. The first is to identify fake users. The other is to detect abnormal statistics in the LDP estimates. The detection results will inform mitigation strategies in the following steps. The post-processing will further reshape the LDP estimates to suppress the attack gain and boost data utility. 

%\sun{add vspace below}

\vspace{2pt}
\textbf{Path 1: \blackcirclea{a} $\Rightarrow$ \blackcircleb{b} $\Rightarrow$ \blackcirclea{c}.} Given the collected LDP reports, our \textit{fake user detection} attempts to identify malicious users (\textbf{step} \blackcircleb{a}).  The corresponding user reports will be subsequently removed from the report collection before being sent to the LDP aggregation (\textbf{step} \blackcircleb{b}). At this point, we consider the collected dataset to be clean without attack influence.  The post-processing methods will be used to further boost the data utility (\textbf{step} \blackcircleb{c}). The experiment shows that our detection outperforms the state-of-the-art in identifying bogus users in MGA and MGA-A attacks. 

\vspace{2pt}
\textbf{Path 2: \blackcirclea{a} $\Rightarrow$ \blackcircleb{d} $\Rightarrow$ \blackcirclea{e} $\Rightarrow$ \blackcircleb{f}.} If detecting fake users is challenging (e.g., under APA or weak attacks close to the baseline attack), the collected data will first be sent to the LDP aggregation function (\textbf{step} \blackcircleb{d}). The estimated result will be further examined by looking for abnormal statistics (\textbf{step} \blackcircleb{e}). Should an attack be recognized, depending on the data practices and available resources, the server either terminates the protocol and recollects data, or post-processes the result to reduce the attack gain and recover utility (\textbf{step} \blackcircleb{f}).

\vspace{2pt}
\textbf{Path 3: \blackcirclea{a} $\Rightarrow$ \blackcircleb{d} $\Rightarrow$ \blackcirclea{e} $\Rightarrow$ \blackcirclea{g}.} This path is identical to \textbf{Path 2} except for the last step, where no attack is detected. In this case, the LDP estimate will be post-processed to increase utility (\textbf{step} \blackcircleb{g}).

\subsubsection{Metrics} \label{sec:metric}
For fake user detection, we measure the $F1$ score while we evaluate the abnormal statistics detection using \textit{accuracy}.

We consider both \textit{utility boosting} and \textit{attack suppression} for LDP post-processing in the presence of data poisoning attacks. We use the widely adopted metric, \textit{mean square error} (MSE) \cite{wang2017locally, wang2019locally} for utility. 
 For attack suppression, frequency gain reduction is an intuitive metric to measure the change of the attack gain after post-processing. However, this metric is inherently limited in interpreting a meaningful comparison of the attack recoverability across various post-processing methods. For instance, a positive gain reduction,  without additional context information, is a weak indicator of how much of the attack has been contained. In other words, it is unclear if the attack can still significantly impact the LDP result. In addition, this metric varies given different target itemset size $r$ and malicious user percentage $\beta$. A more precise measure of attack influence at the per-fake-user and per-item level is desired. In this paper, we present a new metric, \textit{item gain ratio} (\IGR), inspired by prior work \cite{li2024robustness} to address the above issues. \IGR measures a normalized frequency gain change per target item caused by a fake user after post-processing versus the attack gain by the baseline attack. This is because the baseline represents the minimum damage the attacker can always cause in any circumstances. Formally, 
 \begin{align}
    \IGR &= \frac{\sum_{t\in T}\mathbb{E}(\tilde{f}_{t,\textit{recovery}}^{\prime}-\tilde{f}_{t,\textit{before}})}{\sum_{t\in T}\mathbb{E}(\tilde{f}_{t,\textit{base}}-\tilde{f}_{t,\textit{before}})\cdot r}
\end{align}
where $\tilde{f}_{t,\textit{before}}$ is the frequency before attack, $\tilde{f}_{t,\textit{recovery}}^{\prime}$ is the frequency after post-processing, and $\tilde{f}_{t,\textit{base}}$ is the frequency under a baseline attack.

Intuitively, when $\IGR \gg 1/r$, the attack is still considered to be strong compared to the baseline at the per-item level. When it decreases and approaches $1/r$, the efficacy of the attack gets close to that of the baseline attack, which is generally considered no longer a significant threat. When $\IGR\in [0, \frac{1}{r}]$, the attack becomes even weaker than the baseline. A negative \IGR indicates excessively suppressed frequencies of target items, which may lead to deteriorated data utility. 

%-------------------------------------------------------------------------------
\section{Attack Detection}
In this section, we introduce the proposed attack detection methods.

\subsection{Fake User Detection}
In prior work~\cite{cao2021data}, FIAD uses frequent itemset mining to find malicious itemsets by checking if the number of supporting users is greater than a predefined threshold. As a result, all the supporting users for these itemsets will be considered malicious. However, the detection accuracy is largely contingent on how precise we can derive the threshold, which is often challenging in practice. Moreover, the mining process is time-consuming, leading to delayed detection. In this paper, we propose a novel \textit{differential statistical anomaly detection} (\Diffstats) to overcome the above limitations. \Diffstats (in Algorithm~\ref{algorithm_Diffstats})  identifies malicious users by adopting a different strategy that looks at the statistical differences, i.e., $E_{sq}(k)$ and $E_{freq}$, between fake and genuine reports. Our detection includes two steps. First, \Diffstats computes the discrepancy $E_{sq}(k)$ between the observed frequency $O^{k}$ (line 1) of bits set to 1 in the reports of $n$ users and the expected frequency $Y^{k}$ (line 2). It then leverages $E_{sq}(k)$ to determine a candidate fake user set $\mathcal{U}_{s}$ from all users $\mathcal{U}$ (line 6-10). Second, to reduce false positives, \Diffstats calculates $E_{freq}$ for all $\{\mathcal{U}/\mathcal{U}_{sc}\}$ (line 15), where $\mathcal{U}_{sc}$ are subsets of $\mathcal{U}_{s}$ (line 14). Consequently, the smaller $E_{freq}$ is, the more likely that the corresponding $\{\mathcal{U}/\mathcal{U}_{sc}\}$ contains more genuine users and $\mathcal{U}_{sc}$ is malicious (line 16-17). The two steps are performed iteratively for all $k$ until producing the final result of fake users $\mathcal{U}_{f}$. The detailed description of identifying $\mathcal{U}_{sc}$ for reduced false positives is provided in Section~\ref{put_all_together}. In addition, our method is much faster in fake user identification compared to FIAD (see Section~\ref{sec:Result_Diffstats}).

\subsubsection{Frequency Approximation}
\label{sec_theoretical_dist} 
In this section, we first describe deriving the frequency $X$ of the number of ``1'' bits for one user and then extend it to $Y^{k}$ for $n$ users. We take OUE as an example and later discuss OLH, HST, and GRR. 

\begin{algorithm}[t!]
\caption{\textbf{\Diffstats for Fake User Detection}}\label{algorithm_Diffstats}
% \small
\begin{flushleft}
        \textbf{Input:} Users' reports $\mathcal{Y}$.\\
        \textbf{Output:} Fake users $\mathcal{U}_f$.
\end{flushleft}
	\begin{algorithmic}[1]
		\STATE Set $O^k = |\{y^{(j)} \in \mathcal{Y} \mid \sum_{v\in \mathcal{D}} \mathds{1}_{S(y^{(j)})}(v) = k \}|$
        \STATE Set $Y^ k = n \cdot P(X = k)$.
            \STATE Initialize the minimum $E_{freq}$ as $E_{min} = +\infty, \mathcal{U}_{f} = \emptyset$
            % \STATE \xl{Initialize the number of ``1'' set $\mathcal{K} = \{0,1,2,...,d\}$}
            \STATE Set $\mathcal{K} = \{0,1,2,...,d\}$
            % \STATE \sun{Set $\mathcal{E}_{sq} = \{E_{sq}(k)=(O^k - Y^k)^2|k \in \mathcal{K}\}$}
            % \WHILE{$\mathcal{E}_{sq} \neq \emptyset$}
            \WHILE{$\mathcal{K} \neq \emptyset$}
		% \FOR{ $\delta$ in  $(0, \max(E_{sq}))$ $\delta=\min({E_{sq}(k)})$}
            \STATE Let $\delta=\arg\min_{k\in \mathcal{K}}E_{sq}(k)$%{\mathcal{E}_{sq}}$}
            \STATE Update $\mathcal{K} = \{\mathcal{K}\backslash \delta\}$
            \STATE Initialize candidate fake users set $\mathcal{U}_s = \emptyset$
            \FOR{$k \in \mathcal{K}$} 
            \STATE Add users to $\mathcal{U}_s$ whose reports contain $k$-bit ``1''.
            \ENDFOR
            % \STATE \sun{Find the candidate set $U_s$ of fake users} according to Eq. (\ref{suspicious_users})
            % \STATE Compute common support $S$ for $u \in \{U_s\}$
            % \STATE Filter out the $L$ highest frequency items $\mathcal{S}_L$ in $\mathcal{S}$
            % \STATE Initialize all combinations of $\mathcal{S}_L$ using $\mathcal{P}(\mathcal{S}_L)$ according to Eq. \eqref{All_combin}
            \STATE For the top-$L$ supported items $\mathcal{S}_L$ in the common support set $S$, compute all combinations $\mathcal{P}(\mathcal{S}_L)$ in $\mathcal{U}_s$.
            % \STATE \hspace{0.5em} \textbf{For} $s$ in $\mathcal{P}(S_T)$
            \FOR{$s$ in $\mathcal{P}(\mathcal{S}_L)$}
            \STATE Obtain subset $\mathcal{U}_{sc}\subseteq\mathcal{U}_{s}$ that supports $s$.
            \STATE Compute $E_{freq}$ for all users in $\{\mathcal{U}\backslash \mathcal{U}_{sc}\}$.
            \IF{$E_{freq}<E_{min}$}
            \STATE Update $E_{min} = E_{freq}$ and  $\mathcal{U}_f = \mathcal{U}_{sc}$
            \ENDIF
            \ENDFOR
            % \STATE Update $\mathcal{E}_{sq} = \{\mathcal{E}_{sq} \backslash \ \min(\mathcal{E}_{sq})\}$
            \ENDWHILE
		\RETURN $\mathcal{U}_f$
	\end{algorithmic}
\end{algorithm}

For the OUE perturbation, we consider the event of keeping the bit of 1 (i.e., the user's item) a Bernoulli trial, i.e., $X_a \sim \text{Bernoulli}(p)$ and flip the remaining $d-1$ bits independently to 1 with probability $q$. As a result, we may use a binomial distribution for this process, i.e., $X_b \sim B(d-1, q)$. The distribution $X$ of the number of ``1'' bits in a user’s report can be estimated as the sum of two random variables, $X_a$ and $X_b$, following a Poisson binomial distribution $X \sim PB(p_1, p_2, \ldots, p_d)$, where $X = \sum_{i=1}^d X_i$ and each $X_i$ is an independent Bernoulli random variable with expected value $E[X_i] = p_i$. $p_i = p$ if $i = v$, and $p_i = q$ otherwise. We subsequently approximate $X \sim PB(p_1, p_2, \ldots, p_d)$ using  a binomial distribution $X \sim B(d, \tilde{p})$, 
% \begin{equation}
%     X \sim B(d, \tilde{p})
% \end{equation}
where $\tilde{p} = \frac{p + (d-1)q}{d}$ is the mean of $p_i$. The approximated error bound can be derived as $dist(PB,B) \leq (1 - \tilde{p}^{d+1} - \tilde{q}^{d+1})\frac{\sum_{i=1}^d (p_i - \tilde{p})^2}{(d+1) \cdot \tilde{p} \cdot \tilde{q}}$ according to \cite{EHM19917}, where $\tilde{q} = 1 - \tilde{p}$. $dist(PB, B)$ gets close to 0 if and only if the ratio $R_{var}$ of the variance of $PB$ to that of $B$ approaches 1. To see this, for $B(d, \tilde{p})$ and its variance $Var(B) = d \tilde{p}(1 - \tilde{p})$, we have $ R_{var} = \frac{Var(PB)}{Var(B)} = 1 - \frac{(d-1)(p-q)^2}{d^2\tilde{p}(1-\tilde{p})}$. 
%  \begin{align}
%      R_{var} = \frac{Var(PB)}{Var(B)} = 1 - \frac{(d-1)(p-q)^2}{d^2\tilde{p}(1-\tilde{p})}. \nonumber
% \end{align}
Since $d$ is typically large in underlying applications (e.g., $d > 100$ items), $R_{var}$ approaches 1, resulting in $dist(PB, B)$ close to 0. Therefore, we can readily approximate $X\sim B(d, \frac{p + (d-1)q}{d})$ for a single user. Since the perturbations for $n$ users are i.i.d. processes,  the expected frequency $Y^k$ of bits set to ``1'' in the reports of $n$ genuine users is
\begin{equation}\label{theoretical_dist}
    Y^k \approx n \cdot P(X = k),\nonumber
\end{equation}
where $P(X = k) = \binom{d}{k} \tilde{p}^k (1-\tilde{p})^{d-k}$ is the probability mass function of the binomial distribution, given that exactly $k$ bits in a single user's vector are set to 1.

\vspace{2pt}
\textbf{Other Target CFOs Protocols.} 
The expected frequency $Y^k$ serves as a baseline for calculating errors $E_{freq}$ and $E_{sq}$. Under different CFOs, the distribution $X$ varies given distinct encoding and perturbation functions, which thus results in different $Y^k$. In OLH, a malicous user crafts a fake report $y^{(j)}=\langle h^{(j)}, v^{(j)}_h \rangle$ with $\sum_{t \in T} \mathds{1}_{S(y^{(j)})}(t) = r$. OLH requires that the choice of the hash function from $\textbf{H}$ is uniform. This ensures the expected size of the support set of the perturbed value $y^{(j)}$ in the input domain $\mathcal{D}$ is $\frac{d}{g}$. Therefore, $X \sim B(d, \frac{1}{g})$. Likewise, in HST, we consider the number of 1's in the public vector. Assuming an ideal and uniform public vector, the distribution is modeled as $X \sim B(d, \frac{1}{2})$. GRR is more vulnerable to MGA attacks compared to OUE and OLH when $d > (2r - 1)(e^\epsilon - 1) + 3r$~\cite{cao2021data}. Since each user only reports a single value, our fake user detection and prior work do not apply. However, the attack can still be caught by our abnormal statistics detection (see Section~\ref{asd}).

\subsubsection{Quantifying Frequency Discrepancies}\label{sec_dist_discrepancy}

In this subsection, we elaborate on quantifying the errors $E_{sq}$ and $E_{freq}$ to measure the gap between observed frequency $O^{k}$ and expected frequency $Y^{k}$, which helps us identify malicious users.  $E_{sq}$ indicates the error between  $O^k$ and $Y^k$ for a particular $k$ for all $n$ users. Therefore, we have
\begin{equation}\label{fit_error_k} E_{sq}(k) = (O^k - Y^k)^2 \nonumber\end{equation} where $O^k = |\{y^{(j)} \in Y \mid \sum_{v\in \mathcal{D}} \mathds{1}_{S(y^{(j)})}(v) = k \}|$. We define the error $E_{freq}$ as the chi-square statistic in the chi-square goodness-of-fit test \cite{pearson1896vii}. Formally, 
% We adopt the chi-square statistic in the chi-square goodness-of-fit test \cite{pearson1896vii} to define the error $E_{freq}$ as below.
\begin{equation}\label{eq:fit_error}
    E_{freq}(O^k, Y^k) = \sum_k \frac{(O^k - Y^k)^2}{Y^k},
\end{equation}
where the degree of freedom is $k - 1$. Thus, $E_{freq}$ is the sum of $E_{sq}(k)$ for all $k$. Note that our analyses below are generic and apply to OUE, OLH, and HST.

\vspace{2pt}
\textbf{Error Analysis for MGA and MGA-A.} MGA and MGA-A adopt the same attack strategy by setting additional bits up to the expected number of ``1'' in a report of a genuine user. Therefore, the following results apply to both. 

\begin{theorem}\label{errorMGA}
    For $m$ fake users and $k\in[0,d]$, the expected error $\mathbb{E}\left[E_{sq}(k)\right]$ of the MGA and MGA-A is
    \begin{equation}\label{MGA_A_error}
     \mathbb{E}\left[E_{sq}^{MGA}(k)\right]=\left\{\begin{array}{ll}
          m^2\cdot (P(X=k)-1)^2 + Var(O^k_{MGA}), & \mbox{if} \ k = l_g, \nonumber\\
          m^2\cdot (P(X=k))^2 +  Var(O^k_{MGA}), & \mbox{otherwise}, 
     \end{array}
     \right.
\end{equation}
where $l_g = \lfloor p + (d - 1)q \rfloor$ is the expected number of ``1'' in a genuine user's report and $Var(O^k_{MGA}) = (n-m)P(X=k)(1-P(X=k))$ denotes the variance of the observed frequency $O^k_{MGA}$ under the MGA attack.
\end{theorem}

\begin{proof}
   See Appendix \ref{app:proof_errorMGA}
\end{proof}

Given Theorem~\ref{errorMGA}, we further derive $E_{freq}$ in Theorem~\ref{fiterrorMGA}.
\begin{theorem}\label{fiterrorMGA}
The error $E_{freq}(O^k_{MGA}, Y^k)$ between the observed frequency $O^k_{MGA}$ under MGA or MGA-A attack and the expected frequency $Y^k$ is
    \begin{align}\label{eq:fiterrorMGA}
        E_{freq}(O^k_{MGA}, Y^k) &= \frac{m^2}{n}\left(\frac{1}{P(X=l_g)}-1\right) + \frac{(n-m)\cdot d}{n}
    \end{align}
\end{theorem}
\begin{proof}
    See Appendix \ref{app:proof_fiterrorMGA}
\end{proof}

Theorem~\ref{errorMGA} and Theorem~\ref{fiterrorMGA} show that both errors have a quadratic relationship with $m$, i.e., the number of fake users. They are sensitive to the change of $m$, which is desirable. This aligns with our intuition that more fake users generally lead to a stronger attack with more skewed statistics, thus benefiting the detection. In addition, the fake users in MGA and MGA-A set $l_g$ ``1" in their reports while in practice the probability $P(X=l_g)$ for genuine users is low. Therefore, $E_{freq}$ becomes evident in the presence of attacks. The detection performance also relies on the domain size $d$. For simplicity, we approximate $X \sim N(\mu = d\tilde{p}, \sigma = \sqrt{d\tilde{p}(1 - \tilde{p})})$. Since $\mathbb{E}\left[P(X = l_g)\right] = \mathbb{E}\left[f(x = \mu)\right] = \frac{1}{\sigma\sqrt{2\pi}} = \frac{1}{\sqrt{2\pi d\tilde{p}(1 - \tilde{p})}}$, a small $d$ gives rise to a large $P(X = l_g)$ and a small $ E_{freq}$, thus rendering the detection less effective.

\vspace{2pt}
\textbf{Error Analysis for APA.} We analyze the expected error $\mathbb{E}\left[E_{sq}^{APA}(k)\right]$ for APA by Theorem~\ref{errorAPA} below. % provides the expected error $\mathbb{E}\left[E_{sq}^{APA}(k)\right]$ for APA. 

\begin{theorem}\label{errorAPA}
    For m fake users and $k\in[0,d]$, the expected error $\mathbb{E}\left[E_{sq}^{APA}(k)\right]$ under APA attack is 
    \begin{align}
    \mathbb{E}\left[E_{sq}^{APA}(k)\right] = (m \cdot P(X=k) - \omega[k])^2 + Var(O^k_{APA}). \nonumber
\end{align}
$\omega[k] \in [0,m]$ is the number of attacker vectors in which exact $k$ bits are set to 1 and $Var(O^k_{APA}) = (n-m)P(X=k)(1-P(X=k))$ denotes the variance of the observed frequency $O^k_{APA}$ under the APA attack.
\end{theorem}

\begin{proof}
See Appendix \ref{app:proof_errorAPA}. %The proof follows a similar procedure as Theorem 1.
\end{proof}

% \xl{The Expected fit erorr of APA is:}
Theorem~\ref{fiterrorAPA} gives the error $E_{freq}$ for the APA attack.

\begin{theorem}\label{fiterrorAPA}
    % \xl{Based on Lemma~\ref{errorAPA} and Eq. \eqref{fit_error}, the fitting error between the observed frequency $O_{APA}$ under the APA and the expected frequency $Y^k$ is:}
    The error  $E_{freq}(O^k_{APA}, Y^k)$ between the observed frequency $O^k_{APA}$ under APA attack and the expected frequency $Y^k$ is
    \begin{align}
        E_{freq}(O^k_{APA}, Y^k) &= \frac{1}{n}\sum\limits_{\substack{k = 0}}^{d} \frac{(m \cdot P(X=k) - \omega[k])^2}{P(X=k)} + \frac{(n-m)\cdot d}{n}\nonumber
    \end{align}
\end{theorem}
\begin{proof}
    See Appendix \ref{app:proof_fiterrorAPA}
\end{proof}

Unlike in MGA and MGA-A, Theorem~\ref{errorAPA} shows that $E_{sq}^{APA}(k)$ in APA is affected by not only $m$ but also $\omega[k]$. A similar relationship is observed in $E_{freq}$ in Theorem~\ref{fiterrorAPA}. In other words, the attacker can hide traces by adjusting $\omega[k]$ for different $k$ to offset the impact of $m$, which makes the detection challenging in practice. We further theoretically analyze the detection performance under different attacks in Section~\ref{Performance_Analysis}.

\subsubsection{Put All Together} \label{put_all_together}
\Diffstats employs $E_{sq}$ and $E_{freq}$ to identify potential bogus reports, as in Algorithm~\ref{algorithm_Diffstats}. Specifically, $E_{sq}$ is adopted to find fake user candidates, while $E_{freq}$ is subsequently used to minimize false positive rates and refine the final result.

To generate a candidate set $\mathcal{U}_s$ of fake users, we first determine $\delta = \arg\min_{k\in \mathcal{K}}E_{sq}(k)$ for $\mathcal{K}=[0,1,...,k]$. Subsequently, we add users into $\mathcal{U}_s$ whose reports contain $k$-bit 1's for $k \in \mathcal{K} \setminus {\delta}$, as their $E_{sq}(k)$ is relatively large compared to $E_{sq}(\delta)$. However, the derived $\mathcal{U}_s$ may contain genuine users.

%and  the false positive rate is $\frac{(n-m)\cdot \sum_{k = l_g}\binom{d}{k} p^k (1-p)^{d-k}}{n}$ in this case. For MGA-A attack with the default setting in ~\cite{cao2021data} with $n = 100,000$, $\varepsilon = 1$, $m = 5,000$, $d = 1,024$, approximately 2,809 genuine users would be erroneously classified as malicious. 

To reduce false positives, we measure the error $E_{freq}$ of a group of ``genuine'' users   $\{\mathcal{U}\backslash\mathcal{U}_s\}$. The smaller the derived $E_{freq}$,  the higher the likelihood that $\mathcal{U}_s$ is malicious. Note that this counters the intuition of directly measuring the distance between malicious and benign groups. This is because $O^k$ and $Y^k$ are anticipated to characterize statistically similar populations by definition, and in practice, we only know the expected pattern $Y^k$ of genuine users.  

To filter out ``clean'' users, we examine the behavioral differences between attackers and benign users. We observe that the fake users consistently set corresponding bits of the target items to 1 in their reports, which is unlikely among genuine users. %as illustrated in Fig.~\ref{example_Vectors}. 
After determining the common support set $\mathcal{S} = \{S_i\}_{i\in[1,...,d]}$ by summing up the bits at $i$-th position of the user reports, we derive top-$L$ supported items $\mathcal{S}_L$ from $\mathcal{S}$.  For all combinations $\mathcal{P}({\mathcal{S}_L})$ of $\mathcal{S}_L$, \Diffstats identifies a subset $\mathcal{U}_{sc}$ of $\mathcal{U}_s$ and computes $E_{freq}$ for ``clean'' users in $\{\mathcal{U}\backslash \mathcal{U}_{sc}\}$. If the error is smaller than the current minimum, $\mathcal{U}_{sc}$ is likely to be malicious since the attacker aims to maximize the attack gain.   The algorithm will eventually determine a group of users $\mathcal{U}_{f}$ as fake users by iterating through all $k\in \mathcal{K}$. %\sun{can we use common support to directly identify suspicious users?}

\subsubsection{Performance Analysis of \Diffstats}\label{Performance_Analysis} We theoretically analyze the performance of the detection for different attacks. In MGA-A, the attacker can adjust $r$ to evade FIAD detection (e.g., $r\le 2$ in \cite{cao2021data}). Theoretically, a large $r$ helps \Diffstats distinguish honest and malicious users through common support. On the other hand, the candidate set $\mathcal{U}_s$ also facilitates this process. The experiments show that our detection can effectively identify fake users even with a small $r$. We analyze the error relationship $R$ between APA and MGA (MGA-A) in Theorem~\ref{errorComp}.

\begin{theorem}\label{errorComp}
    % For $k = \lfloor p + (d - 1)q \rfloor$, the expected error $\mathbb{E}[E^{APA}_{sq}(k)]$ of APA is \sun{no greater}  than $\mathbb{E}[E^{MGA}_{sq}(k)]$ since
    For $k = \lfloor p + (d - 1)q \rfloor$ the relationship between $\mathbb{E}\left[E^{MGA}_{sq}(k)\right]$ and $ \mathbb{E}\left[E^{APA}_{sq}(k)\right]$ satisfies
    \begin{equation}\label{APA_vs_MGA}
      \left\{\begin{array}{ll}
          \mathbb{E}\left[E^{MGA}_{sq}(k)\right] = \mathbb{E}\left[E^{APA}_{sq}(k)\right], & \mbox{if} \ \omega[k] = m, \\ \\
          \mathbb{E}\left[E^{MGA}_{sq}(k)\right] > \mathbb{E}\left[E^{APA}_{sq}(k)\right], & \mbox{if} \ \omega[k] < m \nonumber
     \end{array}
     \right.
\end{equation}
    % \begin{equation*}
    %     \frac{\mathbb{E}\left[E^{APA}_{sq}(k)\right]}{\mathbb{E}\left[E^{MGA}_{sq}(k)\right]} = \frac{(P(X=k) - \frac{\omega[k]}{m})^2}{(P(X=k) - 1)^2} \xl{\in [0,1]}
    % \end{equation*}
\end{theorem}
\begin{proof}
    See Appendix~\ref{app:proof_errorComp}
\end{proof}

In addition to $m$, $\mathbb{E}[E_{sq}^{APA}(k)]$ is also affected by $\omega[k]$, which is the number of attacker vectors where $k$ bits are set to 1. Adjusting $\omega[k]$ will effectively change the stealthiness of APA. Theorem~\ref{errorComp} shows that when $\omega[k=l_g] = m$, APA is in the worst-case scenario and becomes equivalent to MGA; when $\omega[k] < m$, APA is stealthier than MGA. According to Theorem~\ref{errorAPA}, when $\omega[k] = m \cdot P(X = k)$ for all $k$,  $\mathbb{E}[E_{sq}^{APA}(k)]$ reaches its lower bound $Var(O^k_{APA})$, which represents the optimal attack strategy for APA. In this case, it is difficult to identify suspicious users even if $m$ is large. To address this issue, we present a new detection method below.

\subsection{Abnormal Statistics Detection}\label{asd}
We propose \textit{abnormal statistics detection} (\ASD) to detect the APA attack and the MGA and MGA-A attacks in GRR. The design is inspired by the observation that the sum of the true item counts $\sum_{i=1}^d C_i$ should not exceed $n$; thus the sum of all perturbed counts $\sum_{i=1}^d \tilde{C}_i > n$ may indicate the presence of the attack that promotes the target items by amplifying the corresponding frequencies. However, $\sum\tilde{C}_i$ may exceed $n$ due to random LDP noise, making the na\"ive use of this condition unreliable. To address this challenge, we derive a threshold $\xi$ that divides $\{\tilde{C}_i\}_{i\in[d]}$ into two subsets, $\mathcal{A}=\{\tilde{C}_i|\tilde{C}_i>\xi\}$ and $\mathcal{B}=\{\tilde{C}_i|\tilde{C}_i \le\xi\}$, such that for the items in $\mathcal{A}$, the sum of their perturbed counts should be close to $n$ but still not surpass it. Since the attack gain is unlikely to exist in the lower count domain $\mathcal{B}$, $\sum_{\tilde{C}_i\in \mathcal{A}}\tilde{C}_i > n $ will be a tighter and more sensitive detection condition.  We elaborate on how to find such $\xi$ below.

\subsubsection{Determining $\xi$}\label{Theoretical_ASD} Directly determining $\xi$ as the lower bound of $\mathcal{A}$ is challenging. Instead, we attempt to find $\xi$ as the upper bound of $\mathcal{B}$, which is anticipated to contain much fewer items compared to $\mathcal{A}$. As a result, $\mathcal{A}$ can be derived by excluding $\mathcal{B}$ in the full domain. 

Intuitively, $\mathcal{B}$ should at least contain items whose true frequencies $f_i = 0$, because the expected sum of their perturbed counts would not contribute to $n$. Thus, we are looking for a heuristic $\mathcal{B}=\{\tilde{C}_i|f_i = 0\}$ for $i\in[d]$. It is difficult to derive the precise upper bound $\xi$ of $\mathcal{B}$ without underlying data knowledge, such as true item frequencies. Instead, we aim to find an approximation $\xi(\gamma)$ with confidence level $\gamma$. To this end, we first estimate the distribution of $\tilde{C}_i$. Prior work~\cite{wang2017locally} shows that $\tilde{C}_i \sim N(n\cdot f_i, \ \sigma_i^2)$. %Since $\tilde{C}_i$ are independent, we may derive $\sum_{i=1}^d \tilde{C}_i\sim N(n, \sum_{i=1}^d \sigma^2_i)$.\sun{why we need this?} 
Theorem~\ref{threshold_ASD} gives the upper bound $\xi(\gamma)$ of the set $\mathcal{B}$ below.
\begin{theorem}\label{threshold_ASD}
   For a given confidence level $\gamma$, the upper bound of $\mathcal{B}=\{\tilde{C}_i|f_i = 0\}$ for $i\in[d]$ is
    \begin{equation}
        \xi(\gamma) = Z(\gamma) \cdot \sqrt{\frac{nq(1-q)}{(p-q)^2}} \nonumber
    \end{equation}
\end{theorem}
\begin{figure*}[htbp]
\centering
\includegraphics[scale=0.389]{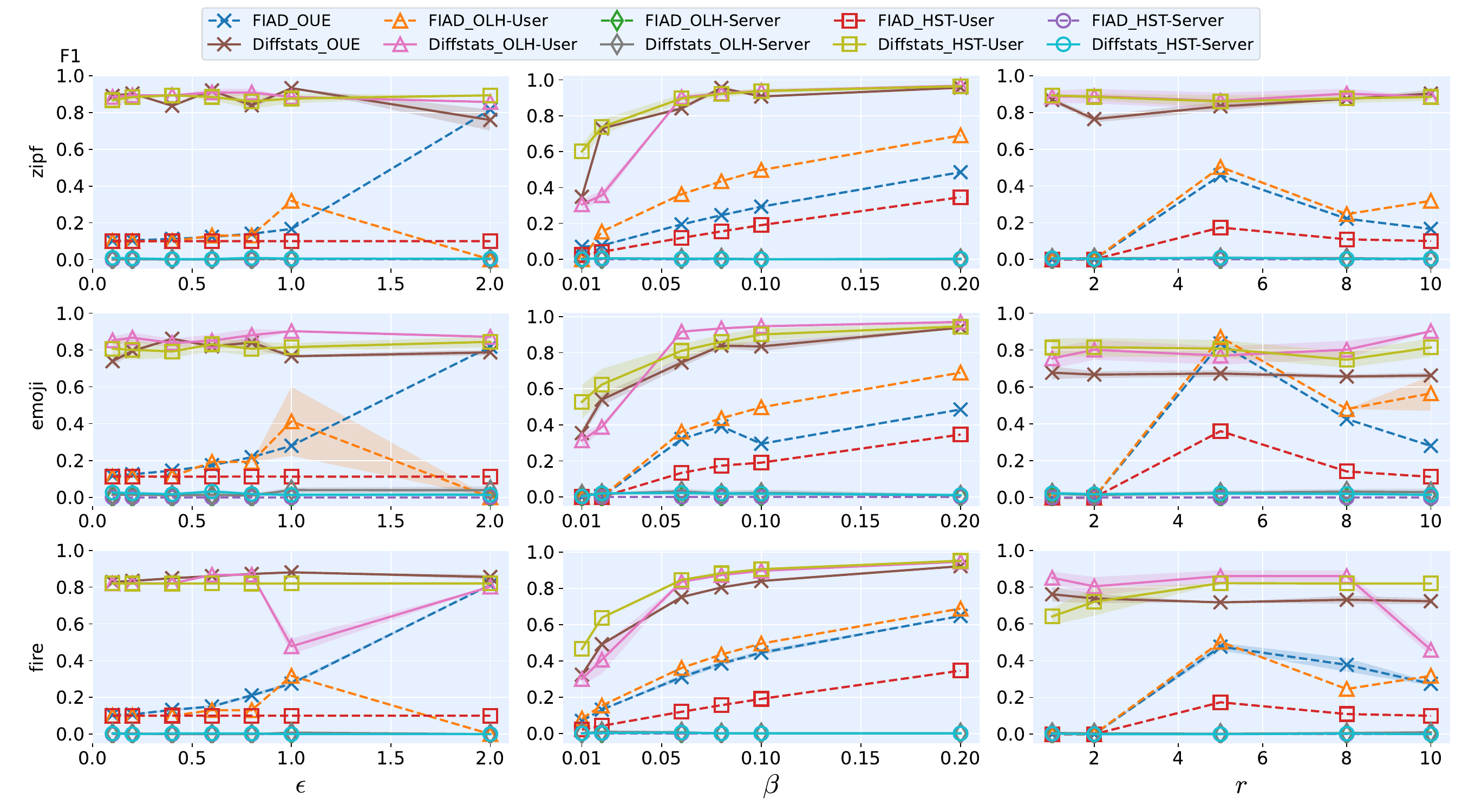}
% \vspace{-5pt}
 \caption{Performance comparison between the proposed \Diffstats and FIAD \cite{cao2021data} against the MGA attack.}
 \label{Fake_user_detection_MGA}
 \Description{Performance comparison between the proposed \Diffstats and FIAD \cite{cao2021data} against the MGA attack.}
  %\vspace{-5pt}
\end{figure*}

\begin{proof}
See Appendix~\ref{app:proof_threshold_ASD}.
\end{proof}

$\xi(\gamma)$ indicates that we have confidence $\gamma$ that all the items in set $\mathcal{B}$ have their true frequencies $f_i = 0$. A larger $\gamma$ also results in a larger $\xi(\gamma)$, which increases the likelihood that all the items with $f_i = 0$ are included in $\mathcal{B}$. On the other hand, $\mathcal{A}$ may also contain zero-frequency items with probability $1-\gamma$. Thus the introduced error $Err$ may cause expected sum $\mathbb{E}[\sum_{\tilde{C}_i \in \mathcal{A}}\tilde{C}_i] > n$ and affect detection accuracy. It is challenging to precisely calculate $Err$ since we need to know the true frequencies $f_i$ (see Appendix~\ref{app:ASD_error}). Alternatively, we approximate the error as
    \begin{equation}
        Err = |\mathcal{B}| \cdot \xi(\gamma)\cdot  (1-\gamma) \nonumber
    \end{equation}
where $|\mathcal{B}|$ is the size of the set $\mathcal{B}$ and $\xi(\gamma)\cdot  (1-\gamma)$ indicates the occurrence of an error event that $\mathcal{A}$ contains an item of $f_i=0$ with probability $1-\gamma$. This error is an approximation of the true error (Eq.~\eqref{true_error} in Appendix~\ref{app:ASD_error}) and represents the upper-bound of the error corresponding to the items with $f_i=0$.

Therefore, we formulate the detection as an optimization problem to find the minimum  $\gamma$ 
% \begin{align}\label{total_error}
%     &\min \ \gamma \nonumber\\
%     &\textrm{s.t.} Err < \lambda \cdot n; 0<\lambda<1
% \end{align}
\begin{equation}
\begin{array}{rcclcl}
\displaystyle \min & \gamma\\
\textrm{s.t.} & Err < \lambda \cdot n\\
&0<\lambda<1    \\ \nonumber
\end{array}
\end{equation}
The condition confines the error below a predefined threshold $\lambda \cdot n$ while ensuring a minimum $\gamma$ to keep $\mathbb{E}[\sum_{\tilde{C}_i\in \mathcal{A}} \tilde{C}_i]$ close to $n$.

\subsubsection{Performance Analysis of \ASD}\label{Performance_Analysis_ASD} 
Since we estimate $\tilde{C}_i \sim N(n\cdot f_i, \ \sigma_i^2)$, a larger $n$ makes the distribution closer to the normal distribution according to the central limit theorem and thus may lead to better detection results. We may also derive $\sum_{i=1}^d \tilde{C}_i\sim N(n, \sum_{i=1}^d \sigma^2_i)$ due to independent $\tilde{C}_i$. We can see that with more items, the variance will increase and thus it is difficult to satisfy $\mathbb{E}[\sum_{\tilde{C}_i \in \mathcal{A}}\tilde{C}_i] \le n$. The detection accuracy is also expected to decline. On the other hand, a weak attack with a large $\epsilon$ and a small $\beta$ will cause the attacked $\tilde{C}_i$ to be close to their expected true count $n\cdot f_i$. Therefore, the detection will be less effective by satisfying the condition $\sum_{\tilde{C}_i \in \mathcal{A}}\tilde{C}_i > n$.

\section{Detection Evaluation}\label{sec_detection_evaluation}
\subsection{Experimental Setup}
\textit{Datasets.} We use one synthetic dataset and two real-world datasets to evaluate the proposed detection methods.
\begin{itemize}%[leftmargin=*]
\item \textbf{zipf.} We synthesize a dataset containing $d = 1,024$ items and $n = 1,000,000$ users satisfying a zipf distribution, where $s = 1.5$.
\item \textbf{emoji.} 
%This dataset contains  a total of $n=218,477$ instances across $d=1,496$ unique emojis~\cite{emoji_dataset}. 
We use $d=1,496$ emojis from ~\cite{emoji_dataset} and treat their average ``sent'' statistics as the number of users, i.e., $n=218,477$.

% \xl{We used $d=1,496$ emojis from ~\cite{emoji_dataset}and treated their average ``sent'' data as average emoji usage data from an emoji keyboard, which reports a total of $n=218,477$ emojis used across $d=1,496$ unique emojis ~\cite{emoji_dataset}.}
\item \textbf{fire.} The dataset contains information on calls for the San Francisco Fire Department \cite{sf_fire_department_2024}. 
We only use the ``Alarms'' records and take each unique unit ID as an item. Therefore, there are $d = 296$ items from $n = 723,090$ users.
%We focus on the "Alarms" category, including $d = 296$ items and $n = 723,090$ users.
\end{itemize} 

The datasets and their pre-processing details, along with the source code, are provided at \textbf{\url{https://github.com/Marvin-huoshan/MDPA_LDP/}}. %\sun{github}
\vspace{2pt}
\textit{Settings.} The experiments were conducted on a server with Ubuntu 22.04.5 LTS, 2$\times$ AMD EPYC 9554 CPU, and 768GB RAM.
We use a default $\epsilon=1$ for LDP and $\beta=0.05$, $r=10$ for all attacks unless specified in the paper. We also set $r'=4$ by default for APA for balanced stealthiness and efficacy of the attack.
We follow the optimal APA attack strategy by setting $\omega[k] = \lfloor m \cdot P(X = k)\rfloor$ for all $k$ as discussed in Section~\ref{Performance_Analysis}. In OLH-User, we allow the attacker to find an optimal hash function $h^{(j)}$ that maps all target items in $T$ to $v_h^{(j)}$.

For \Diffstats detecting fake users, we select $L=6$ by default and set the hyperparameter $\lambda=0.02$ in \ASD detection.  %\sun{Do we use the default parameters or specify them in the caption?}

\vspace{2pt}
\textit{Metrics.} We measure $F1$ score for our method \Diffstats to correctly identify fake users while minimizing false positives and false negatives. The results are an average of 10 trials. We measure the detection \textit{accuracy} for \ASD over a total of 40 instances including 20 attack and 20 no-attack cases. 

In the experiment, we compare the proposed \Diffstats with the state-of-the-art fake user detection, FIAD \cite{cao2021data} and only shows the performance of \ASD since no similar method exists for detecting MGA and MGA-A attacks on GRR and APA attack.

\subsection{Results for \Diffstats}\label{sec:Result_Diffstats}
The detection results against MGA are shown in Figure ~\ref{Fake_user_detection_MGA} and those against MGA-A in Figure~\ref{Fake_user_detection_MGA-A}. A higher $F1$ score indicates better detection performance, and shaded regions represent 95\% confidence intervals (CI). Since the majority of the measured CIs are narrow, i.e., typically $[0.02,0.05]$, which makes them difficult to see in the figures.

\vspace{2pt}
\textbf{Detection for MGA.} In general, our experiment shows that \Diffstats outperforms FIAD by a large margin across all datasets with varying parameters. Both detection methods become less effective under the server setting of OLH and HST. This is because those protocols are naturally more robust to the attacks, which is consistent with prior results \cite{li2024robustness}. In addition, we have the following key observations.
\begin{itemize}%[leftmargin=*]
\item \Diffstats performs almost constantly with $F1$ greater than 0.8 versus $F1$ of FIAD below 0.4 for a common setting of $\epsilon \le 1$ across all datasets. 
\vspace{2pt}
\item We observe a performance drop of our method in the fire dataset of a small domain size $d$ when $\epsilon=1$ in OLH-User. This aligns with our analysis of the error $E_{freq}$ in Eq.~\eqref{eq:fiterrorMGA} that a small $d$ leads to a loss in detection performance. At the same time, the attack becomes weak when it is difficult to find a hash function that covers all target items as $\epsilon$ grows. However, performance bounces back quickly with a larger $\epsilon$ since fewer genuine users are included in the common support in \Diffstats, thus reducing false positives. In contrast, the performance of FIAD for OLH-User is significantly affected given a large $\epsilon$. 

\vspace{2pt}
\item We observe that the $F1$ of FIAD first increases and then reduces with growing $\epsilon$ for OLH-User with all three datasets. Given a smaller $\epsilon$, though a stronger attack presents for it is easier for more fake users to find hash functions to cover all target items, more LDP noise also results in higher false positives, which is more prominent in detection performance (i.e., low $F1$ score). As the negative impact of LDP noise diminishes much faster than the difficulty of obtaining ideal hash functions with increasing $\epsilon$, the combined effect makes the $F1$ score of FIAD first peak at about $\epsilon = 1$ and then decrease rapidly. This phenomenon is consistent across all datasets, with a wider CI in the emoji dataset. This is because the number of fake users in the emoji is much smaller than that in the other two, leading to a more ``volatile'' measurement.
\vspace{2pt}
\item While our method performs stably in OUE across all $\epsilon$ with a narrow $95\%$ CI typically ranging from $0.02$ to $0.05$, FIAD only becomes more effective as $\epsilon$ increases. With less LDP noise, the identified frequent itemsets in FIAD are likely to include the target items. 

\vspace{2pt}
\item The performance (i.e., $F1$ and CI) of both detection methods improves for stronger attacks with more malicious users (i.e., larger $\beta$) except for the server settings. Compared to FIAD, our method stays sensitive to small $\beta$ and grows much faster to an $F1$ score of above 0.9 as $\beta$ increases. 
\vspace{2pt}
\item \Diffstats is insensitive to changes in the target set size $r$, consistently showing superior performance. The $F1$ score drops when $r$ becomes larger for OLH-User as the attacker fails to find a qualified hash function. On the other hand, FIAD cannot identify fake users when $r \leq 2$.  It performs best at $r = 5$ and worsens as $r$ further grows. This is because a relatively large target set helps FIAD filter out malicious items in general. However, excessively large $r$ leads to a higher false positive rate in FIAD. 
\end{itemize}

\begin{figure*}[htbp]
\centering
% \vspace{-5pt}
\includegraphics[scale=0.389]{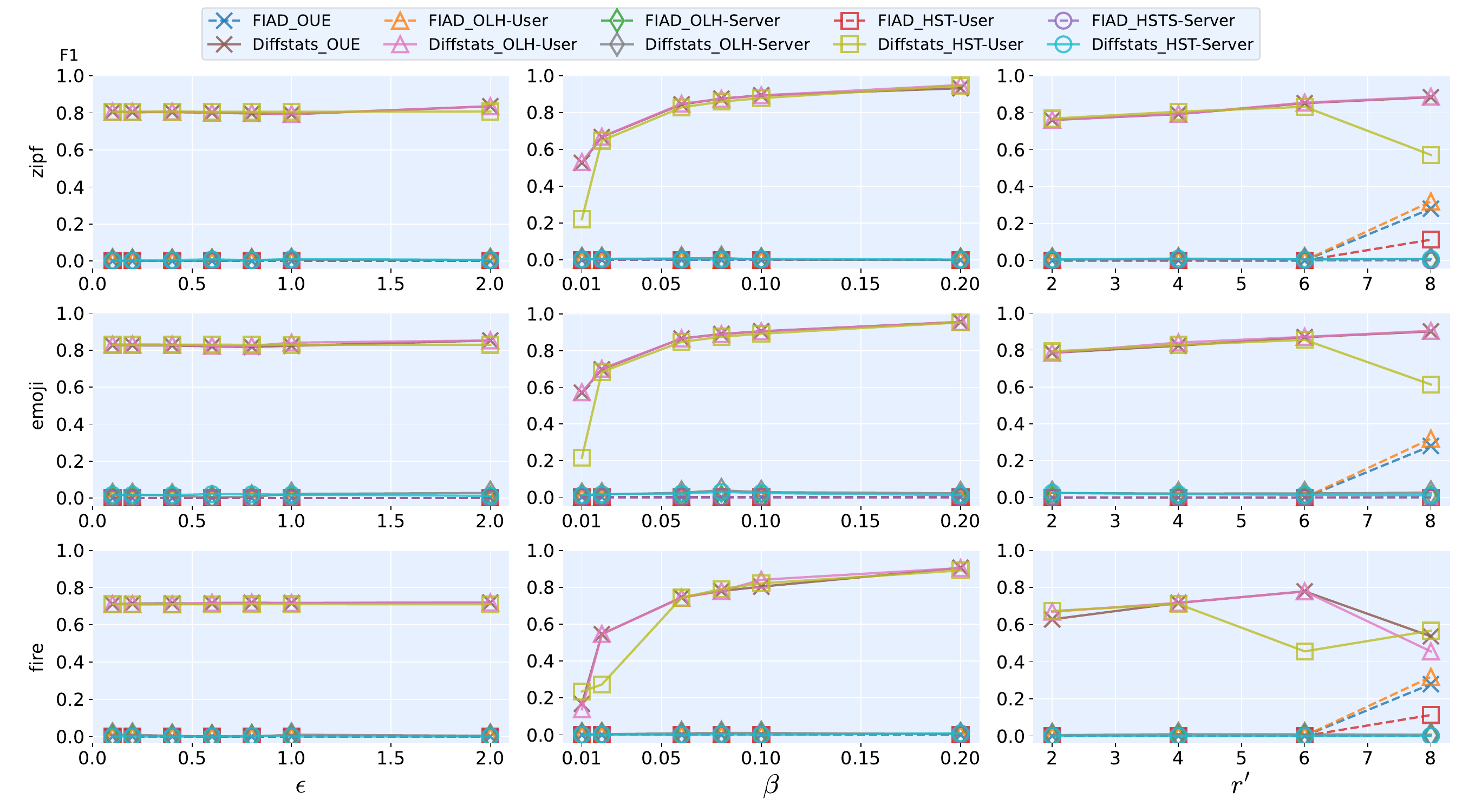}
 \caption{Performance comparison between the proposed \Diffstats and FIAD \cite{cao2021data} against the MGA-A attack.}
 \label{Fake_user_detection_MGA-A}
 % \vspace{-10pt}
 % \Description{Performance comparison between the proposed \Diffstats and FIAD \cite{cao2021data} against the MGA-A attack.}
 % \vspace{5pt}
\end{figure*}

\vspace{2pt}
\textbf{Detection for MGA-A.} For MGA-A attack, Figure~\ref{Fake_user_detection_MGA-A} shows that:  
\begin{itemize}%[leftmargin=*]
    \item The performance gap between our method and FIAD becomes even clearer against the MGA-A attack. \Diffstats shows a constantly high and stable $F1$ score while FIAD fails to detect fake users in most cases. 
    % \item Under different values of $\epsilon$, \Diffstats maintains robust detection against MGA-A attacks. However, FIAD, which relies on frequent itemset mining, fails to detect fake users when MGA-A splits the target set $T$ of length $r$ into smaller subsets of length $r'$.
    \vspace{2pt}
    \item Consistent with the MGA case, our detection is more effective when a strong attack with a large $\beta$ is present. 
    \vspace{2pt}
    \item Compared to other protocols, \Diffstats in HST-User is more subject to the changes of $r'$. Especially when $r'$ gets larger, higher false positives may appear as more items exist in the common support. 
    \vspace{2pt}
    \item \Diffstats performs better in zipf and emoji since the fire dataset has a small domain size $d$, making the detection more challenging.
    % \item Under different values of $\beta$, \Diffstats improves its detection performance as $\beta$ increases, consistent with its behavior against MGA. In contrast, FIAD remains ineffective at identifying fake users across varying $\beta$ values.
    % \item For parameter $r'$, representing the group size in MGA-A, \Diffstats consistently outperforms FIAD in detection effectiveness.
    % \begin{itemize}[leftmargin=*]
        % \item For the HST-User protocol, which uses a binary hash to encode each item, it is effectively similar to the OLH protocol with $g=2$ when $\epsilon < 0.5$. Compared to the other protocols depicted (with $\epsilon$ fixed at 1), HST-User exhibits higher randomness. Consequently, when $r'$ is larger (leading to more items in the Common Support), this results in a higher false positive rate, as more genuine users are incorrectly detected.
        % \item On the Fire dataset, \Diffstats exhibits poorer performance, consistent with previous discussions. Under both MGA and MGA-A attacks, a smaller domain size $d$ results in a lower fitting error, which makes detection more challenging.
        % \item For the FIAD method, it is evident that it can only achieve a certain level of detection capability when $r' > 6$.
    % \end{itemize}
    % \item Under different protocol Server settings, none of the detection methods achieve effective results. 
\end{itemize}

We further evaluate our approach and FIAD for APA with the optimal strategy described in Section~\ref{Performance_Analysis}. The experiment confirms our theoretical analysis that APA is much stealthier than existing data poisoning attacks of similar \IGR and evades detection in all three datasets, i.e., $F1=0$ with $r'=4$, $\epsilon\in [0.1, 1]$ and $\beta\in [0.01, 0.1]$. In addition, neither FIAD nor ours supports GRR. 

\vspace{2pt}
\textbf{Time Cost.} We measure the time cost for \Diffstats and FIAD in Table~\ref{tab:running_time}. 
%in Appendix~\ref{app:diffstats}. 
Our experiments demonstrate up to about 120$\times$ detection speedup versus FIAD, with a 1-hour timeout\footnote{$37.04\%$ of FIAD instances were not recorded due to the timeout.}. Thus, \Diffstats is more friendly to time-sensitive applications.

\begin{table}[bthp]
%\vspace{-5pt}
\setlength\tabcolsep{1pt}
\footnotesize
%\small
\centering
\caption{Time cost in seconds of FIAD and \Diffstats. Default $\epsilon = 1$, $r = 10$, and $\beta = 0.05$. ``--'' indicates passing the 1-hour timeout.}
\label{tab:running_time}
\begin{tabular}{cccccccccccc}
\hline
\multirow{2}{*}{Dataset} & \multirow{2}{*}{} & \multirow{2}{*}{} & \multicolumn{1}{c}{OUE} & & \multicolumn{1}{c}{OLH-User} & & \multicolumn{1}{c}{OLH-Server} & & \multicolumn{1}{c}{HST-User} & & \multicolumn{1}{c}{HST-Server} \\
 &  &  & FIAD~/~Ours & & FIAD~/~Ours & & FIAD~/~Ours & & FIAD~/~Ours & & FIAD~/~Ours \\ \hline
\multicolumn{1}{l|}{} & \multicolumn{1}{l}{} & $0.1$ & --~/~117 & & --~/~116 & & --~/~103 & & --~/~117 & & --~/~104\\
\multicolumn{1}{l|}{} & $\epsilon$ & $0.5$ & --~/~122 & & --~/~124 & & --~/~107 & & --~/~115 & & --~/~102\\
\multicolumn{1}{l|}{} & \multicolumn{1}{l}{} & $1$ & --~/~129 & & 2,441~/~130 & & 2,389~/~109 & & --~/~114 & & --~/~101\\ \cline{2-12} 
\multicolumn{1}{l|}{} &  & $1\%$ & 2,677~/~118 & & 2,336~/~120 & & 2,270~/~105 & & --~/~106 & & --~/~100\\
\multicolumn{1}{l|}{zipf} & $\beta$ & $5\%$ & --~/~128 & & 2,368~/~131 & & 2,357~/~107 & & --~/~117 & & --~/~100\\
\multicolumn{1}{l|}{} &  & $10\%$ & --~/~140 & & 2,531~/~146 & & 2,432~/~115 & & --~/~129 & & --~/~103\\ \cline{2-12} 
\multicolumn{1}{l|}{} &  & $1$ & 2,586~/~123 & & 2,190~/~126 & & 2,326~/~111 & & --~/~113 & & --~/~101\\
\multicolumn{1}{l|}{} & $r$ & $5$ & 2,826~/~127 & & 3,294~/~130 & & 2,390~/~115 & & --~/~119 & & --~/~101\\
\multicolumn{1}{l|}{} &  & $10$ & --~/~126 & & 2,355~/~129 & & 2,430~/~111 & & --~/~117 & & --~/~99\\ \hline

\multicolumn{1}{l|}{} & \multicolumn{1}{l}{} & $0.1$ & 3,491~/~29 && --~/~28 && --~/~25 && --~/~28 && --~/~24\\
\multicolumn{1}{l|}{} & $\epsilon$ & $0.5$ & 2,230~/~29 && 1,761~/~27 && 1,876~/~25 && --~/~27 && --~/~24\\
\multicolumn{1}{l|}{} & \multicolumn{1}{l}{} & $1$ & 1,171~/~30 & & 1,045~/~28 && 1,077~/~25 && --~/~27 && --~/~24\\ \cline{2-12} 
\multicolumn{1}{l|}{} &  & $1\%$ & 1,206~/~27 && 1,062~/~27 && 1,075~/~25 && --~/~25 && --~/~25\\
\multicolumn{1}{l|}{emoji} & $\beta$ & $5\%$ & 1,169~/~28 && 1,030~/~30 && 1,056~/~24 && --~/~27 && --~/~25\\
\multicolumn{1}{l|}{} &  & $10\%$ & 3,073~/~31 && 1,007~/~32 && 1,072~/~26 && --~/~30 && --~/~23\\ \cline{2-12} 
\multicolumn{1}{l|}{} &  & $1$ & 1,214~/~28 && 976~/~31 && 1,052~/~25 && --~/~27 && --~/~24\\
\multicolumn{1}{l|}{} & $r$ & $5$ & 1,187~/~29 & & 1,037~/~30 && 1,070~/~25 && --~/~27 && --~/~24\\
\multicolumn{1}{l|}{} &  & $10$ & 1,190~/~29 && 1,054~/~30 && 1,109~/~25 && --~/~27 && --~/~24\\ \hline

\multicolumn{1}{l|}{} & \multicolumn{1}{l}{} & $0.1$ & 1,043~/~109 && --~/~106 && 589~/~99 && 2,414~/~106 && 587~/~101\\
\multicolumn{1}{l|}{} & $\epsilon$ & $0.5$ & 544~/~110 && --~/~112 && 289~/~108 && 1,786~/~108 && 592~/~99\\
\multicolumn{1}{l|}{} & \multicolumn{1}{l}{} & $1$ & 341~/~118 && 228~/~122 && 182~/~113 && 1,986~/~109 && 577~/~99\\ \cline{2-12} 
\multicolumn{1}{l|}{} &  & $1\%$ & 225~/~118 && 209~/~113 && 176~/~112 && 1,831~/~101 && 589~/~96\\
\multicolumn{1}{l|}{fire} & $\beta$ & $5\%$ & 337~/~122 && 230~/~123 && 176~/~111 && 1,869~/~106 && 587~/~98\\
\multicolumn{1}{l|}{} &  & $10\%$ & 458~/~127 && 277~/~129 && 191~/~114 && 1,792~/~115 && 590~/~99\\ \cline{2-12} 
\multicolumn{1}{l|}{} &  & $1$ & 195~/~120 && 170~/~122 && 170~/~111 && 557~/~106 && 586~/~97\\
\multicolumn{1}{l|}{} & $r$ & $5$ & 201~/~117 && 307~/~122 && 184~/~115 && 728~/~107 && 609~/~97\\
\multicolumn{1}{l|}{} &  & $10$ & 396~/~125 && 238~/~124 && 180~/~112 && 1,435~/~107 & & 616~/~98\\ \hline
\end{tabular}
%\vspace{-5pt}
\end{table}

\subsection{Results for \ASD}\label{sec:Results_ASD}

We assess the detection accuracy of the proposed \ASD method under challenging circumstances where malicious users cannot be identified under APA attack (Table~\ref{tab:ACC_table_ASD}), and MGA and MGA-A attacks on GRR (Table~\ref{tab:ACC_table_ASD_GRR}). %\sun{remove Since the detection performance is similar across three datasets, we only show the results for zipf.}
\begin{table}[htbp]
\centering
% \vspace{-2pt}
\caption{Detection accuracy with $95\% $ CI of \ASD against APA. We set $\beta=0.1$ and $\epsilon=0.5$ by default.}
\label{tab:ACC_table_ASD}
\begin{tabular}{c|cccccccc}
\toprule
\multirow{1}{*}{Dataset} & \multirow{1}{*}{} & \multirow{1}{*}{} & \multicolumn{1}{c}{OUE} & \multicolumn{1}{c}{OLH-User} & \multicolumn{1}{c}{HST-User} \\
\midrule
\multirow{8}{*}{zipf} 
& \multirow{4}{*}{$\epsilon$} & 0.1 & 1.00~{\scriptsize [0.91, 1]} & 1.00~{\scriptsize [0.91, 1]} & 1.00~{\scriptsize [0.91, 1]}\\
& & 0.5 & 1.00~{\scriptsize [0.912, 1]} & 1.00~{\scriptsize [0.91, 1]} & 1.00~{\scriptsize [0.91, 1]}\\
& & 0.8 & 1.00~{\scriptsize [0.91, 1]} & 1.00~{\scriptsize [0.91, 1]} & 1.00~{\scriptsize [0.91, 1]}\\
& & 1 & 1.00~{\scriptsize [0.91, 1]} & 1.00~{\scriptsize [0.91, 1]} & 1.00~{\scriptsize [0.91, 1]}\\\noalign{\vskip 1pt}\cline{2-6}\noalign{\vskip 1pt}
& \multirow{4}{*}{$\beta$} & 1\% & 0.98~{\scriptsize [0.88, 0.99]} & 0.88~{\scriptsize [0.73, 0.95]} & 0.53~{\scriptsize [0.36, 0.68]}\\
& & 5\% & 1.00~{\scriptsize [0.91, 1]} & 1.00~{\scriptsize [0.91, 1]} & 1.00~{\scriptsize [0.91, 1]}\\
& & 7.5\% & 1.00~{\scriptsize [0.91, 1]} & 1.00~{\scriptsize [0.91, 1]} & 1.00~{\scriptsize [0.91, 1]}\\
& & 10\% & 1.00~{\scriptsize [0.91, 1]} & 1.00~{\scriptsize [0.91, 1]} & 1.00~{\scriptsize [0.91, 1]}\\
\midrule
\multirow{8}{*}{emoji} 
& \multirow{4}{*}{$\epsilon$} & 0.1 & 1.00~{\scriptsize [0.91, 1]} & 1.00~{\scriptsize [0.91, 1]} & 1.00~{\scriptsize [0.91, 1]}\\
& & 0.5 & 1.00~{\scriptsize [0.91, 1]} & 1.00~{\scriptsize [0.91, 1]} & 1.00~{\scriptsize [0.91, 1]}\\
& & 0.8 & 1.00~{\scriptsize [0.91, 1]} & 1.00~{\scriptsize [0.91, 1]} & 1.00~{\scriptsize [0.91, 1]}\\
& & 1 & 1.00~{\scriptsize [0.91, 1]} & 1.00~{\scriptsize [0.91, 1]} & 1.00~{\scriptsize [0.91, 1]} &\\ \noalign{\vskip 1pt}\cline{2-6}\noalign{\vskip 1pt} 
& \multirow{4}{*}{$\beta$} & 1\% & 0.50~{\scriptsize [0.33, 0.66]} & 0.50~{\scriptsize [0.33, 0.66]} & 0.50~{\scriptsize [0.33, 0.66]}\\
& & 5\% & 1.00~{\scriptsize [0.91, 1]} & 1.00~{\scriptsize [0.91, 1]} & 1.00~{\scriptsize [0.91, 1]}\\
& & 7.5\% & 1.00~{\scriptsize [0.91, 1]} & 1.00~{\scriptsize [0.91, 1]} & 1.00~{\scriptsize [0.91, 1]}\\
& & 10\% & 1.00~{\scriptsize [0.91, 1]} & 1.00~{\scriptsize [0.91, 1]} & 1.00~{\scriptsize [0.91, 1]}\\
\midrule
\multirow{8}{*}{fire} 
& \multirow{4}{*}{$\epsilon$} & 0.1 & 1.00~{\scriptsize [0.91, 1]} & 1.00~{\scriptsize [0.91, 1]} & 1.00~{\scriptsize [0.91, 1]}\\
& & 0.5 & 1.00~{\scriptsize [0.91, 1]} & 1.00~{\scriptsize [0.91, 1]} & 1.00~{\scriptsize [0.91, 1]}\\
& & 0.8 & 1.00~{\scriptsize [0.91, 1]} & 1.00~{\scriptsize [0.91, 1]} & 1.00~{\scriptsize [0.91, 1]}\\
& & 1 & 1.00~{\scriptsize [0.91, 1]} & 0.98~{\scriptsize [0.88, 0.99]} & 1.00~{\scriptsize [0.91, 1]}\\\noalign{\vskip 1pt}\cline{2-6}\noalign{\vskip 1pt}
& \multirow{4}{*}{$\beta$} & 1\% & 0.50~{\scriptsize [0.33, 0.66]} & 0.65~{\scriptsize [0.48, 0.79]} & 0.53~{\scriptsize [0.36, 0.68]}\\
& & 5\% & 1.00~{\scriptsize [0.91, 1]} & 1.00~{\scriptsize [0.91, 1]} & 1.00~{\scriptsize [0.91, 1]}\\
& & 7.5\% & 1.00~{\scriptsize [0.91, 1]} & 1.00~{\scriptsize [0.91, 1]} & 1.00~{\scriptsize [0.91, 1]}\\
& & 10\% & 1.00~{\scriptsize [0.91, 1]} & 1.00~{\scriptsize [0.91, 1]} & 1.00~{\scriptsize [0.91, 1]}\\
\bottomrule
\end{tabular}
% \vspace{-3pt}
\end{table}

\vspace{2pt}
\textbf{Detection for APA.} Table~\ref{tab:ACC_table_ASD} shows that our \ASD can effectively detect the APA attack with various $\epsilon$ irrespective of the underlying LDP protocols. In general, a small $\beta$ results in a decrease in accuracy since the attack also becomes weaker. We do not observe a performance drop until $\beta=0.01$, which indicates the effectiveness of \ASD even when the attacker controls a small portion of users in the system. The detection becomes stable with a narrower CI as $\beta$ increases. %\xl{Moreover,as $\beta$ increases, the 95\% confidence‑interval widths shrink, indicating that \ASD’s detection performance becomes increasingly stable under stronger attacks.}
Our analysis in Section~\ref{Performance_Analysis_ASD} also shows that the performance improves with increasing $n$ and decreases with a larger $d$.  This is confirmed by our experiment, i.e., the best result is observed in zipf with $n=1,000,000$ and $d=1,024$, while the worst occurs in emoji with $n=218,477$ and $d=1,496$. %(Table~\ref{tab_app:ACC_table_ASD_emoji&fire} in Appendix~\ref{app:overall_emoji&fire}). 
With a small $\beta$, the detection performs better in OUE and OLH-User than in HST-User. This is because the attack under the same setting on HST-User is weaker than the other two.

\vspace{2pt}
\textbf{Detection for MGA and MGA-A in GRR.} Similar to the detection against APA attack, Table~\ref{tab:ACC_table_ASD_GRR} shows that \ASD consistently achieves high accuracy ($\ge 88\%$) for all settings. Unlike in APA, the detection is still effective even with small $\beta$ but is more subject to change of $\epsilon$. In other words, a small $\epsilon$ results in low detection accuracy.  The reason is that the variance of GRR is greater than that of other protocols, which introduces more LDP noise and negatively affects the detection performance.

\vspace{2pt}
\textbf{Time Cost.} Our experiment shows that \ASD is very fast, consistently under 1$s$ across all tested settings (see Table~\ref{tab_app:running_time_ASD_APA} and Table~\ref{tab_app:running_time_ASD_GRR}).

\begin{table}[tbp]
\centering
% \scriptsize
% \vspace{-2pt}
\caption{Detection accuracy with $95\% $ CI of \ASD for GRR. We set $\beta=0.1$, $\epsilon=0.5$ and $r=10$ by default.}%Accuracy of \ASD on GRR protocol on zipf dataset. During $\epsilon$ variation, $\beta$ is fixed to 0.1; during $\beta$ variation, $\epsilon$ is fixed to 0.5; during $r$ or $r^{\prime}$ variation, $\epsilon$ is fixed to 0.5 $\beta$ is fixed to 0.1.}
\label{tab:ACC_table_ASD_GRR}
\begin{tabular}{c|ccc|ccccc}
\toprule
\multirow{1}{*}{Dataset} & \multirow{1}{*}{} & \multicolumn{1}{c}{MGA} & \multicolumn{1}{c|}{} & \multirow{1}{*}{} & \multicolumn{1}{c}{MGA-A} & \multirow{1}{*}{}  \\
\midrule
\multirow{12}{*}{zipf} 
& \multirow{4}{*}{$\epsilon$} & 0.1 & 0.93~{\scriptsize [0.79, 0.98]} & \multirow{4}{*}{$\epsilon$} & 0.1 & 0.88~{\scriptsize [0.73, 0.95]}\\
& & 0.5 & 1.00~{\scriptsize [0.91, 1]} & & 0.5 & 1.00~{\scriptsize [0.91, 1]}\\
& & 0.8 & 1.00~{\scriptsize [0.91, 1]} & & 0.8 & 1.00~{\scriptsize [0.91, 1]}\\
& & 1 & 1.00~{\scriptsize [0.91, 1]} & & 1 & 1.00~{\scriptsize [0.91, 1]}\\
\noalign{\vskip 1pt}\cline{2-7}\noalign{\vskip 1pt} 
& \multirow{4}{*}{$\beta$} & 1\% & 1.00~{\scriptsize [0.91, 1]} & \multirow{4}{*}{$\beta$} & 1\% & 1.00~{\scriptsize [0.91, 1]}\\
& & 5\% & 1.00~{\scriptsize [0.91, 1]} & & 5\% & 1.00~{\scriptsize [0.91, 1]}\\
& & 7.5\% & 1.00~{\scriptsize [0.91, 1]} & & 7.5\% & 1.00~{\scriptsize [0.91, 1]}\\
& & 10\% & 1.00~{\scriptsize [0.91, 1]} & & 10\% & 1.00~{\scriptsize [0.91, 1]}\\\noalign{\vskip 1pt}\cline{2-7}\noalign{\vskip 1pt}  
& \multirow{4}{*}{$r$} & 1 & 1.00~{\scriptsize [0.91, 1]} & \multirow{4}{*}{$r^\prime$} & 2 & 1.00~{\scriptsize [0.91, 1]}\\
& & 2 & 1.00~{\scriptsize [0.91, 1]} & & 4 & 1.00~{\scriptsize [0.91, 1]}\\
& & 5 & 1.00~{\scriptsize [0.91, 1]} & & 6 & 1.00~{\scriptsize [0.91, 1]}\\
& & 10 & 1.00~{\scriptsize [0.91, 1]} & & 8 & 1.00~{\scriptsize [0.91, 1]}\\
\midrule
\multirow{12}{*}{emoji} 
& \multirow{4}{*}{$\epsilon$} & 0.1 & 0.93~{\scriptsize [0.79, 0.98]} & \multirow{4}{*}{$\epsilon$} & 0.1 & 0.95~{\scriptsize [0.83, 0.99]}\\
& & 0.5 & 1.00~{\scriptsize [0.91, 1]} & & 0.5 & 1.00~{\scriptsize [0.91, 1]}\\
& & 0.8 & 1.00~{\scriptsize [0.91, 1]} & & 0.8 & 1.00~{\scriptsize [0.91, 1]}\\
& & 1 & 1.00~{\scriptsize [0.91, 1]} & & 1 & 1.00~{\scriptsize [0.91, 1]}\\
\noalign{\vskip 1pt}\cline{2-7}\noalign{\vskip 1pt} 
& \multirow{4}{*}{$\beta$} & 1\% & 1.00~{\scriptsize [0.91, 1]} & \multirow{4}{*}{$\beta$} & 1\% & 1.00~{\scriptsize [0.91, 1]}\\
& & 5\% & 1.00~{\scriptsize [0.91, 1]} & & 5\% & 1.00~{\scriptsize [0.91, 1]}\\
& & 7.5\% & 1.00~{\scriptsize [0.91, 1]} & & 7.5\% & 1.00~{\scriptsize [0.91, 1]}\\
& & 10\% & 1.00~{\scriptsize [0.91, 1]} & & 10\% & 1.00~{\scriptsize [0.91, 1]}\\\noalign{\vskip 1pt}\cline{2-7}\noalign{\vskip 1pt}  
& \multirow{4}{*}{$r$} & 1 & 1.00~{\scriptsize [0.91, 1]} & \multirow{4}{*}{$r^\prime$} & 2 & 1.00~{\scriptsize [0.91, 1]}\\
& & 2 & 1.00~{\scriptsize [0.91, 1]} & & 4 & 1.00~{\scriptsize [0.91, 1]}\\
& & 5 & 1.00~{\scriptsize [0.91, 1]} & & 6 & 1.00~{\scriptsize [0.91, 1]}\\
& & 10 & 1.00~{\scriptsize [0.91, 1]} & & 8 & 1.00~{\scriptsize [0.91, 1]}\\
\midrule
\multirow{12}{*}{fire}
& \multirow{4}{*}{$\epsilon$} & 0.1 & 1.00~{\scriptsize [0.91, 1]} & \multirow{4}{*}{$\epsilon$} & 0.1 & 1.00~{\scriptsize [0.91, 1]}\\
& & 0.5 & 1.00~{\scriptsize [0.91, 1]} & & 0.5 & 1.00~{\scriptsize [0.91, 1]}\\
& & 0.8 & 1.00~{\scriptsize [0.91, 1]} & & 0.8 & 1.00~{\scriptsize [0.91, 1]}\\
& & 1 & 1.00~{\scriptsize [0.91, 1]} & & 1 & 1.00~{\scriptsize [0.91, 1]}\\
\noalign{\vskip 1pt}\cline{2-7}\noalign{\vskip 1pt} 
& \multirow{4}{*}{$\beta$} & 1\% & 1.00~{\scriptsize [0.91, 1]} & \multirow{4}{*}{$\beta$} & 1\% & 1.00~{\scriptsize [0.91, 1]}\\
& & 5\% & 1.00~{\scriptsize [0.91, 1]} & & 5\% & 1.00~{\scriptsize [0.91, 1]}\\
& & 7.5\% & 1.00~{\scriptsize [0.91, 1]} & & 7.5\% & 1.00~{\scriptsize [0.91, 1]}\\
& & 10\% & 1.00~{\scriptsize [0.91, 1]} & & 10\% & 1.00~{\scriptsize [0.91, 1]}\\\noalign{\vskip 1pt}\cline{2-7}\noalign{\vskip 1pt}  
& \multirow{4}{*}{$r$} & 1 & 1.00~{\scriptsize [0.91, 1]} & \multirow{4}{*}{$r^\prime$} & 2 & 1.00~{\scriptsize [0.91, 1]}\\
& & 2 & 1.00~{\scriptsize [0.91, 1]} & & 4 & 1.00~{\scriptsize [0.91, 1]}\\
& & 5 & 1.00~{\scriptsize [0.91, 1]} & & 6 & 1.00~{\scriptsize [0.91, 1]}\\
& & 10 & 1.00~{\scriptsize [0.91, 1]} & & 8 & 1.00~{\scriptsize [0.91, 1]}\\
\bottomrule
\end{tabular}
%\vspace{-5pt}
\end{table}

\section{Attack Recovery of LDP Post-processing}
Post-processing was previously intended to remove excessive LDP noise to boost utility in a compliant environment without attack concerns. Depending on underlying data tasks and available resources, service providers may not be able to recollect data when the \ASD detection result is positive. Therefore, it is crucial to select a post-processing method that can reconstruct as many characteristics of the original data from polluted LDP estimates as possible. In this section, we empirically study the data recoverability of state-of-the-art LDP post-processing methods introduced in Section~\ref{existing_post_methods}, including Norm-Sub and Base-Cut \cite{wang2019locally} for regular no-attack scenarios and the one specifically designed for attack recovery, i.e., Normalization \cite{cao2021data} and LDPRecover \cite{sun2024ldprecover}. In addition, we propose a new post-processing method, \textit{robust segment normalization} (\RSN), which strikes a balance between robustness and accuracy. We first describe the design of \RSN below.

\begin{table}[tb]

\centering
\small
\caption{Time cost (ms) of \ASD against APA. Default $\epsilon = 0.5$, $r = 10$, and $\beta = 0.1$.}
\label{tab_app:running_time_ASD_APA}
% \vspace{-5pt}
\begin{tabular}{c|cccccccc}
\toprule
\multirow{1}{*}{Dataset} & \multirow{1}{*}{} & \multirow{1}{*}{} & \multicolumn{1}{c}{OUE} & \multicolumn{1}{c}{OLH-User} & \multicolumn{1}{c}{HST-User} \\
\midrule
\multirow{6}{*}{zipf} 
& \multirow{3}{*}{$\epsilon$} & 0.1 & 539 & 546 & 554\\
& & 0.5 & 375 & 370 & 356\\
& & 1 & 360 & 369 & 348 &\\ \noalign{\vskip 2pt}\cline{2-6}\noalign{\vskip 2pt}
& \multirow{3}{*}{$\beta$} & 1\% & 542 & 539 & 525\\
& & 5\% & 363 & 366 & 347\\
& & 10\% & 349 & 351 & 348\\
\midrule
\multirow{6}{*}{emoji} 
& \multirow{3}{*}{$\epsilon$} & 0.1 & 592 & 592 & 615\\
& & 0.5 & 451 & 428 & 486\\
& & 1 & 416 & 421 & 411 &\\ \noalign{\vskip 2pt}\cline{2-6}\noalign{\vskip 2pt}
& \multirow{3}{*}{$\beta$} & 1\% & 597 & 591 & 735\\
& & 5\% & 424 & 411 & 411\\
& & 10\% & 411 & 410 & 411\\
\midrule
\multirow{6}{*}{fire} 
& \multirow{3}{*}{$\epsilon$} & 0.1 & 514 & 557 & 631\\
& & 0.5 & 283 & 294 & 290\\
& & 1 & 244 & 249 & 241 &\\ \noalign{\vskip 2pt}\cline{2-6}\noalign{\vskip 2pt} 
& \multirow{3}{*}{$\beta$} & 1\% & 449 & 454 & 481\\
& & 5\% & 274 & 272 & 287\\
& & 10\% & 348 & 318 & 286\\
\bottomrule
\end{tabular}
\end{table}

\subsection{Robust Segment Normalization}\label{RSN_analysis}
In Norm-Sub, $\Delta$ introduces positive bias for low frequencies and negative bias for high frequencies. This approach may reduce the frequencies of all non-target items to 0 under a substantial attack; normalization reduces the attack influences by scaling down all frequency estimates multiplicatively and introduces significant negative errors to high frequencies. 

Our RSN is inspired by existing methods and aims to control the incurred bias more precisely. We observe that the perturbation errors within high-frequency estimates $\mathcal{H}$ are relatively small compared to their true frequencies. Thus, a minimal error adjustment is needed there. On the other hand, we want to find a region $\mathcal{L}$ that contains $\tilde{C}_i$ whose corresponding frequency after LDP may become negative.  We only add $\Delta$ to this region instead of the full domain in Norm-Sub to minimize the impact of bias on high-frequency estimates. Since $\tilde{C}_i\sim N(n\cdot f_i,\sigma_i^2)$ (see Section~\ref{Theoretical_ASD}), we may find a $f_i$ whose minimum perturbed count $\tilde{C}_i = 0$ via $n \cdot f_i - 2 \cdot \sigma_i = 0$.  $\mathcal{L}$ can be approximated by finding the items with perturbed counts less than $n \cdot f_i + 2 \cdot \sigma_i$. \RSN centers on recalibrating estimates within $\mathcal{L}$ since only those frequencies can contribute to negative values after LDP aggregation. This, therefore, reduces the impact of bias on high-frequency estimates. 
$\Delta$ can be derived by  
 $\sum_{\tilde{C}_i\in \mathcal{L}} \max(\tilde{C}_i + \Delta, 0) + \sum_{\tilde{C}_j\in \mathcal{H}} \max(\tilde{C}_j, 0) = \sum{\tilde{C}_i}$.
 In other words, we only adjust low frequencies to become non-negative while keeping high-frequency estimates unchanged. To further maintain a \textit{consistency} condition, the final result will be 
$\tilde{f}_i' = \frac{\tilde{C}_i}{\sum{\tilde{C}_i}}$.
We use multiplicative rather than additive correction to avoid the case where a non-target item in $\mathcal{H}$ is subtracted to 0. In this way, \RSN maintains the relative ratios among the high-frequency values, which are desired for many applications.

\begin{table}[tb]

\centering
\small
\caption{Time cost (ms) of \ASD for GRR. Default $\epsilon = 0.5$, $r = 10$, and $\beta = 0.1$.}
\label{tab_app:running_time_ASD_GRR}
\begin{tabular}{c|ccc|ccccc}
\toprule
\multirow{1}{*}{Dataset} & \multirow{1}{*}{} & \multicolumn{1}{c}{MGA} & \multicolumn{1}{c|}{} & \multirow{1}{*}{} & \multicolumn{1}{c}{MGA-A} & \multirow{1}{*}{}  \\
\midrule
\multirow{9}{*}{zipf} 
& \multirow{3}{*}{$\epsilon$} & 0.1 & 530 & \multirow{3}{*}{$\epsilon$} & 0.1 & 619\\
& & 0.5 & 380 & & 0.5 & 381\\
& & 1 & 372 & & 1 & 378\\\noalign{\vskip 1pt}\cline{2-7}\noalign{\vskip 1pt}
& \multirow{3}{*}{$\beta$} & 1\% & 530 & \multirow{3}{*}{$\beta$} & 1\% & 543\\
& & 5\% & 357 & & 5\% & 359\\
& & 10\% & 355 & & 10\% & 356\\\noalign{\vskip 1pt}\cline{2-7}\noalign{\vskip 1pt} 
& \multirow{3}{*}{$r$} & 1 & 621 & \multirow{3}{*}{$r^\prime$} & 2 & 533\\
& & 5 & 363 & & 6 & 362\\
& & 10 & 363 & & 8 & 363\\
\midrule
\multirow{9}{*}{emoji} 
& \multirow{3}{*}{$\epsilon$} & 0.1 & 568 & \multirow{3}{*}{$\epsilon$} & 0.1 & 564\\
& & 0.5 & 421 & & 0.5 & 418\\
& & 1 & 425 & & 1 & 418\\\noalign{\vskip 1pt}\cline{2-7}\noalign{\vskip 1pt} 
& \multirow{3}{*}{$\beta$} & 1\% & 576 & \multirow{3}{*}{$\beta$} & 1\% & 572\\
& & 5\% & 403 & & 5\% & 403\\
& & 10\% & 402 & & 10\% & 405\\\noalign{\vskip 1pt}\cline{2-7}\noalign{\vskip 1pt}
& \multirow{3}{*}{$r$} & 1 & 571 & \multirow{3}{*}{$r^\prime$} & 2 & 581\\
& & 5 & 400 & & 6 & 403\\
& & 10 & 397 & & 8 & 400\\
\midrule
\multirow{9}{*}{fire} 
& \multirow{3}{*}{$\epsilon$} & 0.1 & 469 & \multirow{3}{*}{$\epsilon$} & 0.1 & 461\\
& & 0.5 & 302 & & 0.5 & 303\\
& & 1 & 295 & & 1 & 295\\\noalign{\vskip 1pt}\cline{2-7}\noalign{\vskip 1pt}
& \multirow{3}{*}{$\beta$} & 1\% & 471 & \multirow{3}{*}{$\beta$} & 1\% & 542\\
& & 5\% & 289 & & 5\% & 284\\
& & 10\% & 289 & & 10\% & 282\\\noalign{\vskip 1pt}\cline{2-7}\noalign{\vskip 1pt}
& \multirow{3}{*}{$r$} & 1 & 455 & \multirow{3}{*}{$r^\prime$} & 2 & 473\\
& & 5 & 285 & & 6 & 290\\
& & 10 & 282 & & 8 & 291\\
\bottomrule
\end{tabular}
\end{table}

\begin{figure*}[htbp]
\centering
 % \vspace{-5pt}
\includegraphics[scale=0.35]{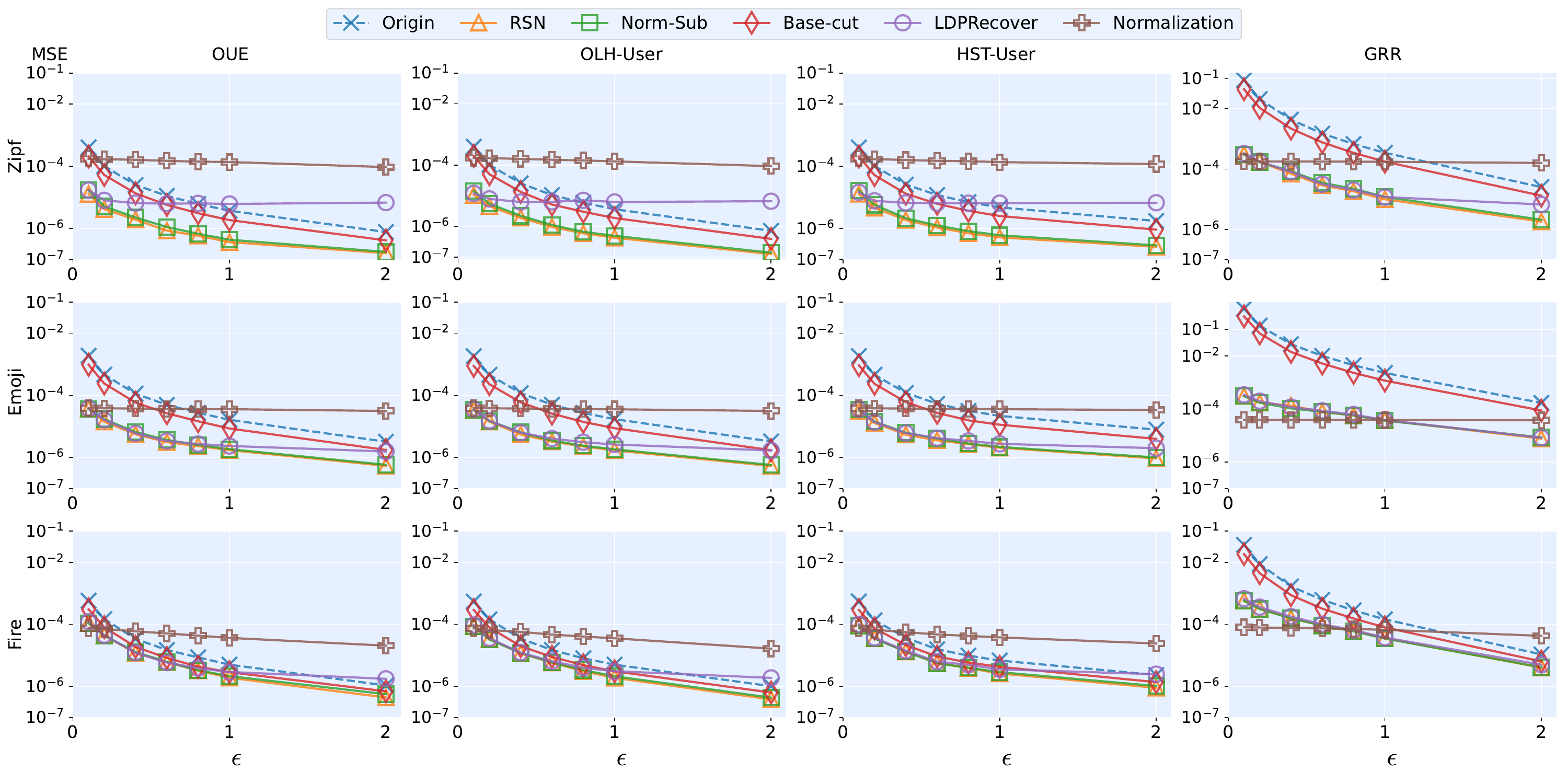}
 \caption{Utility boost of different post-processing methods without attack.}
 \label{Recovery_noattack_MSE}
 \Description{Utility boost of different post-processing methods without attack.}
   %\vspace{-5pt}
\end{figure*}

\begin{figure*}[htbp]
\centering
\includegraphics[scale=0.35]{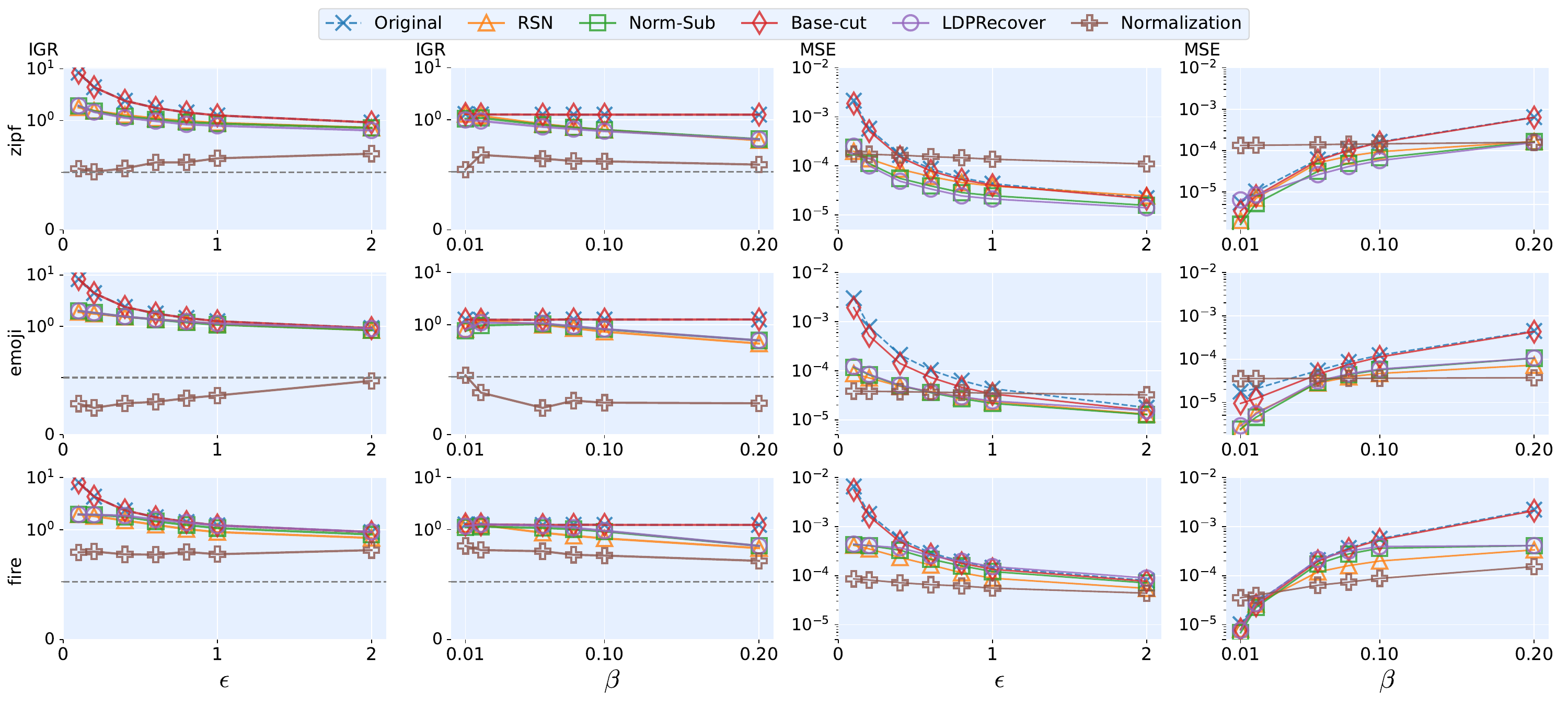}
 \caption{Utility recoverability of different post-processing methods under APA attack on OUE. }
 % The default $\beta=0.05$ and $\epsilon=1$.}
 \label{Recovery_APA_IGR&MSE_OUE}
 \Description{Recoverability of different post-processing methods under APA attack on OUE.}
 % \vspace{5pt}
\end{figure*}

% \begin{figure*}[h]
% \centering
% \includegraphics[scale=0.35]{figure/MSE_Recovery.pdf}
%  \caption{MSE results of different Post processing methods under APA with different parameter variations. During $\epsilon$ variation, $\beta$ is fixed to 0.05; during $\beta$ variation, $\epsilon$ is fixed to 1.}
%  \label{Recovery_APA_MSE}
% \end{figure*}

\begin{figure*}[htbp]
\centering
% \vspace{-5pt}
\includegraphics[scale=0.35]{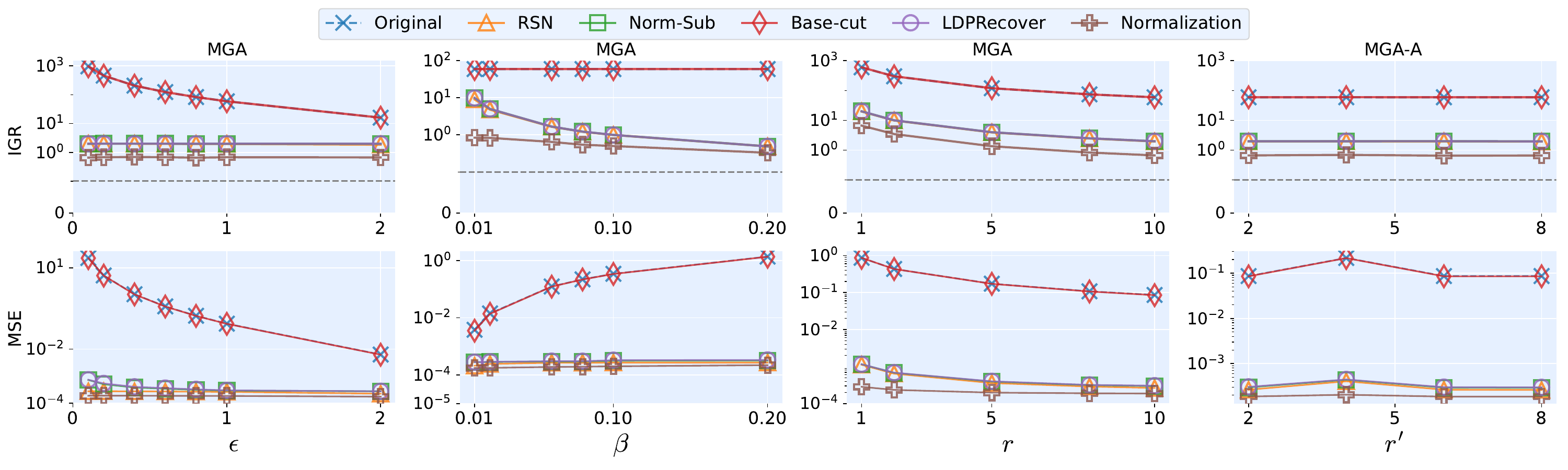}
 \caption{Utility recoverability of different post-processing methods under MGA and MGA-A on GRR in zipf.}
 \label{Recovery_GRR_IGR&MSE_zipf}
 % \Description{\sun{Utility recoverability of different post-processing methods under MGA and MGA-A on GRR in zipf.}}

\end{figure*}

 \begin{figure*}[htbp]
 \centering
 \includegraphics[scale=0.35]{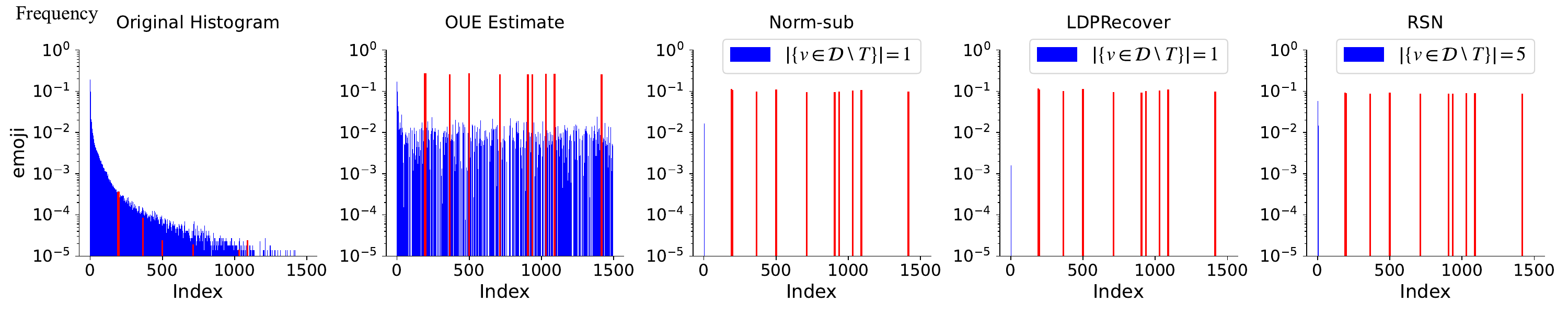}
   \caption{Comparison of histograms 
   after Norm-Sub, LDPRecover, and RSN. The target item frequencies are highlighted in red. $\epsilon = 0.5, r = 10$ and $\beta = 0.05$.}
   \label{hist_post}
   % \Description{Comparison of histograms 
   % after Norm-Sub, LDPRecover, and RSN. The target item frequencies are highlighted in red. $\epsilon = 0.5, r = 10$ and $\beta = 0.05$.}
 \end{figure*}

% \begin{figure*}[h]
% \centering
% \includegraphics[scale=0.35]{figure/MSE_Recovery_GRR.pdf}
%  \caption{MSE results for different post-processing methods under the MGA and MGA-A attacks in the GRR protocol, evaluated across various parameter settings.}
%  \label{Recovery_GRR_MSE}
% \end{figure*}

\subsection{Evaluation}
We empirically study the attack resilience and utility boost of state-of-the-art post-processing methods. The experimental setting remains the same as the detection evaluation in Section~\ref{sec_detection_evaluation}. The results are an average of 10 trials.

\subsubsection{Utility Boost after Fake User Removal} We first consider the situation where MGA and MGA-A attacks are applied to OUE, OLH and HST. The attack influence can be effectively minimized by our detection and thus the underlying data is considered to be clean. We also consider GRR without attack for a comprehensive evaluation of major CFOs.
We show the results for different datasets in Figure~\ref{Recovery_noattack_MSE}.

Overall, our \RSN performs best by exhibiting the lowest MSE in most settings.  LDPRecover and Normalization underperform the regular post-processing approaches, which is expected. The performance of the attack-centered methods is less affected by $\epsilon$ compared to other methods whose MSE decreases as $\epsilon$ grows. It is interesting to see that for small $\epsilon$ LDPRecover performs better than Base-Cut recommended in prior work \cite{wang2019locally}. 

In GRR, Normalization shows lower MSE than other methods when $\epsilon$ is small, which is more obvious with emoji and fire. This is because the high variance of GRR and large noise significantly affect the accuracy of the LDP result, on which the approaches, such as Norm-Sub, rely for utility optimization, while Normalization largely ignores. This is evidenced by consistent MSE values (i.e., about $10^{-4}$) of Normalization for all LDP protocols. 

\subsubsection{Attack Recovery} We consider the utility recoverability of the post-processing methods for the APA attack as well as MGA and MGA-A on the GRR protocol.  

\vspace{2pt}
\textbf{Results for APA Attack.} 
As the results are consistent across different CFO protocols, we only discuss OUE in Figure~\ref{Recovery_APA_IGR&MSE_OUE} and leave others in Figure~\ref{Recovery_APA_IGR&MSE_OLHU} and \ref{Recovery_APA_IGR&MSE_HSTU} in Appendix~\ref{app:recovery}. 

For attack gain suppression, the experimental results are mixed for resilience-optimized methods. On the one hand, Normalization shows lower \IGR closer to the baseline given various $\epsilon$ and $\beta$.  On the other hand, LDPRecover does not show the expected advantages in reducing attack gain. Its performance is on par with that of Norm-Sub and \RSN in zipf, while it underperforms \RSN in emoji and fire. 

For utility boosting, despite superior performance in attack gain reduction, Normalization does not produce lower MSE than other methods due to its coarse-grained bias control. Rather, it again shows the context-agnostic nature of utility recovery with steady MSE values, rendering it an unreliable method. Since Norm-Sub, LDPRecover, and \RSN are all based on the \textit{consistency}, they exhibit similar performance in the experiment while \RSN exhibits lower MSE in the fire dataset.

\vspace{2pt}
\textbf{Results for MGA and MGA-A attack.} Here we only show MGA and MGA-A with varying $r'$ in zipf in Figure~\ref{Recovery_GRR_IGR&MSE_zipf} due to similar results in other settings (see Figure~\ref{Recovery_GRR_IGR} and \ref{Recovery_GRR_MSE} in Appendix~\ref{app:recovery}). %\sun{need to describe the figures in the appendix.}

Normalization outperforms other methods in reducing attack gain in all settings for GRR. LDPRecover, Norm-Sub and \RSN are similar in mitigating attack influence. Regarding MSE, all methods are close with slightly better performance of Normalization. Base-Cut shows the worst in both attack resilience and data accuracy. As $r^\prime$ increases in MGA-A, both \IGR and MSE remain stable for all methods, because each fake user in GRR can support only one target item in their report with probability $1/r$, independent of $r'$.%\xl{This is because, under the GRR protocol, each fake user can support only one target item in their report. With a fixed $r$, the probability of supporting each item in the target set remains $1/r$, regardless of variations in $r^\prime$, making the metrics independent of $r^\prime$.}

\subsubsection{Summary} It is interesting to see that the regular post-processing methods (e.g., Norm-Sub) perform similarly to state-of-the-art attack-focused methods (e.g., LDPRecover) in data recovery. It is speculated that the \textit{consistency} condition plays a crucial role, not only benefiting utility but also attack suppression. When there is no attack or the attack influence has been minimized (e.g., fake users can be identified or in the server setting of OLH and HST), our method \RSN is recommended for utility boosting. When a strong attack is present in OUE, OLH-User, and HST-User, \RSN and Norm-Sub are both preferred methods for data recovery. In general, Normalization is not recommended since it is not adaptive to the underlying data and LDP settings except when large errors are expected (e.g., GRR with large $d$ and small $\epsilon$). Our additional experiment in Figure~\ref{hist_post} shows that given similar MSE, \RSN retains more original high-frequency items compared to Norm-Sub and LDPRecover. As a result, \RSN would be a preferred post-processing method in general for balanced attack mitigation and data recovery. %\sun{check the font type for MSE. They are not consistent in the paper. May use plain text MSE}%\sun{need more references}

% In summary, we evaluate five state-of-the-art post-processing methods designed to enhance attack resilience and boost utility across various datasets, tasks, and parameter settings. Based on the findings, we provide the following guidelines for selecting and utilizing post-processing methods.

% The best choice for a post-processing method depends on two key factors: the presence of an attack and the type of LDP protocol used. Whether or not there is an attack can be accurately determined using the \ASD method, and the type of LDP protocol used is readily available information. This allows for informed selection of the most suitable post-processing method based on the specific scenario:
% \begin{itemize}[leftmargin=*]
%     \item If there is no attack, use \RSN for the best utility and robustness across all scenarios.
%     \item If there is an attack and the LDP protocol is OUE, OLH-User, or HST-User, use Norm-sub or \RSN for effective suppression and utility balance.
%     \item If there is an attack and the LDP protocol is GRR, use Normalization.
% \end{itemize}
% \subsection{Evaluation of proposed post}
%\input{Related_work}
\section{Discussion} 
\vspace{2pt}
\textbf{Extension to Heavy Hitter Identification.} Prior study \cite{cao2021data} shows that PEM protocol for heavy hitter identification \cite{bassily2017practical,bassily2015local,wang2019locally_hh} is also vulnerable to MGA attack. As a result, the target items could be falsely recognized as the top-$k$. Our detection methods can help identify attack behaviors. Specifically, since PEM interactively applies OLH for partial vector perturbation, we may adopt \Diffstats to detect malicious users for each iteration. However, as the size of the encoded domain in PEM is smaller than that of the original item domain, it may affect the detection performance as discussed in Section~\ref{sec_dist_discrepancy}. Likewise, \ASD should also be applicable but could be less effective due to the focus on top-$k$ items in PEM, which prevents full domain aggregation analysis.

\vspace{2pt}
\noindent\textbf{Extension to Numerical Data.} The recent work \cite{li2024robustness} recommended the SW protocol over binning-based CFOs for numerical distribution estimation thanks to the more robust performance of SW against the data poisoning attack and better results of their proposed detection method on SW (i.e., higher AUC values for attack existence). Our research can be extended to enhance the practical robustness of the binning-based CFO protocols.  Specifically, CFOs with binning apply, for example, OUE or OLH, to estimate a histogram and adopt \textit{consistency} post-processing to enforce the distribution characteristics. Since the targeted attack in \cite{li2024robustness} is a variant of MGA, our proposed detection methods can be adapted to detect the attack and identify malicious users. Along with our post-processing method, we may significantly enhance the corrupted data recovery for vulnerable binning-based CFOs and thus provide reliable alternatives for numerical data in hostile environments. 

\section{Conclusion}
Data poisoning attacks pose an imminent threat to current LDP implementations and urge the research community to rethink the security implications in privacy-enhancing technologies. Our newly discovered attack strategies unveil the dynamics of the threat landscape. To mitigate the emerging adversaries, we propose novel detection methods with high detection accuracy and low overhead. In addition, our new post-processing method and follow-up analysis of existing approaches reveal key design principles and highlight the importance of LDP post-processing in attack resilience, which we believe will benefit the research in the future.

% about the attack resilience of various methods and reveals key principles that guide our own design. 

% In this paper, we propose novel defenses to help address the pressing issue. We first expand existing knowledge of attacker capabilities by revealing a new attack with much fewer attack traces. Given diverse attack strategies, we propose innovative detection methods. They either significantly outperform the existing or enable attack detection in new challenging scenarios. We 

% present a new LDP post-processing method and empirically study attack suppression and utility boost of our approach along with other popular post-processing schemes. The generated knowledge provides new insights into the effective mitigation of the emerging data manipulation threat to LDP protocols.  
\section{Acknowledgments}
We would like to thank the Shepherd and anonymous reviewers for their insightful comments and guidance. This paper was supported in part by NSF grants CNS-2238680, CNS-2207204, and CNS-2247794.

%-------------------------------------------------------------------------------
\bibliography{section/cite}
\bibliographystyle{section/ACM-Reference-Format}

\begin{appendices}

\section{Proof of Theorem \ref{errorMGA}}
\label{app:proof_errorMGA}

\begin{proof}
\begin{align}
    \mathbb{E}\left[E_{sq}^{MGA}(k)\right]&= \mathbb{E}\left[(O^k_{MGA} - Y^k)^2\right] \nonumber \\
    &=\mathbb{E}\left[(O^k_{MGA})^2\right]- 2Y^k\mathbb{E}\left[O^k_{MGA}\right] + (Y^k)^2\nonumber \\
    &= (\mathbb{E}\left[(O^k_{MGA})^2\right]-\mathbb{E}\left[O^k_{MGA}\right]^2) \nonumber \\ &+(\mathbb{E}\left[O^k_{MGA}\right]^2 - 2Y^k\mathbb{E}\left[O^k_{MGA}\right] + (Y^k)^2)\nonumber \\
    &= Var(O^k_{MGA}) + (\mathbb{E}[O^k_{MGA}]-Y^k)^2 \nonumber
\end{align}
 
Let $W^k_{benign}\sim B(n-m, \ P(X=k))$ denote the frequency of bits set to 1 in the reports of $n-m$ benign users and $W^k_{fake}$ be the frequency of bits set to 1 in the reports of $m$ fake users. $W^k_{fake}=m$ if $k=l_g$ and $W^k_{fake}=0$ if $k\neq l_g$. Thus, the observed frequency is 
$O^k_{MGA} = W^k_{benign} + W^k_{fake}$. Note that $Y^k = nP(X=k)$.

Given the independence between $W_{benign}^k$ and $W_{fake}^k$, when $k=l_g$, we have 
$\mathbb{E}\left[O^k_{MGA}\right] = \mathbb{E}[W^k_{benign}] + m = (n-m)P(X=k) + m$ \ and 
$Var(O^k_{MGA}) = Var(W^k_{benign}) = (n-m)P(X=k)(1-P(X=k))$.
Therefore,
\begin{align}
    &\mathbb{E}\left[E_{sq}^{MGA}(k)\right] = Var(O^k_{MGA}) + (\mathbb{E}[O^k_{MGA}]-Y^k)^2 \nonumber \\
    &= (n-m)P(X=k)(1-P(X=k)) \nonumber \\
    &+((n-m)P(X=k)+m-nP(X=k))^2 \nonumber \\
    &= (n-m)P(X=k)(1-P(X=k)) + m^2\cdot(P(X=k)-1)^2. \nonumber
\end{align}
When $k\neq l_g$, since $W^k_{fake} = 0$, we have 
$\mathbb{E}\left[O^k_{MGA}\right] = \mathbb{E}[W^k_{benign}] = (n-m)P(X=k)$ and
$Var(O^k_{MGA}) = Var(W^k_{benign}) = (n-m)P(X=k)(1-P(X=k))$. Therefore,
\begin{align}
    &\mathbb{E}\left[E_{sq}^{MGA}(k)\right]= Var(O^k_{MGA}) + (\mathbb{E}[O^k_{MGA}]-Y^k)^2 \nonumber \\
    &= (n-m)P(X=k)(1-P(X=k)) \nonumber \\
    &+ ((n-m)P(X=k)-nP(X=k))^2 \nonumber \\
    &= (n-m)P(X=k)(1-P(X=k)) + m^2\cdot(P(X=k))^2. \nonumber
\end{align}
\end{proof}

\section{Proof of Theorem \ref{fiterrorMGA} \label{app:proof_fiterrorMGA}}
\begin{proof}
\begin{align}
    E_{freq}(O^k_{MGA}, Y^k) &= \sum_k\frac{(O^k_{MGA} - Y^k)^2}{Y^k}\nonumber\\
    &= \frac{(O^{l_g}_{MGA}-Y^{l_g})^2}{Y^{l_g}} + \sum\limits_{\substack{k = 0 \\ k \neq l_g}}^{d} \frac{(O^k_{MGA} - Y^k)^2}{Y^k}\nonumber
\end{align}
According to Theorem~\ref{errorMGA}, when $k = l_g$, we have 
\begin{align}
    &\frac{(O^{l_g}_{MGA}-Y^{l_g})^2}{Y^{l_g}} \nonumber\\
    &= \frac{m^2\cdot(P(X=l_g)-1)^2 + (n-m)P(X=l_g)(1-P(X=l_g))}{n\cdot P(X=l_g)}.\nonumber
\end{align}
When $k \neq l_g$, 
\begin{align}
    &\frac{(O^k_{MGA} - Y^k)^2}{Y^k} \nonumber\\
    &= \frac{m^2\cdot(P(X=k))^2+(n-m)P(X=k)(1-P(X=k))}{n\cdot P(X=k)} \nonumber\\
    &= \frac{m^2\cdot P(X=k)+(n-m)(1-P(X=k))}{n}.\nonumber
\end{align}
According to Eq.~\eqref{eq:fit_error} we have
\begin{align}
     & E_{freq}(O^k_{MGA}, Y^k) \nonumber\\
     &= \frac{m^2\cdot(P(X=l_g)-1)^2 + (n-m)P(X=l_g)(1-P(X=l_g))}{n\cdot P(X=l_g)} \nonumber\\
     &+ \sum\limits_{\substack{k = 0 \\ k \neq l_g}}^{d} \frac{m^2\cdot P(X=k)+(n-m)(1-p(X=k))}{n} \nonumber \\
     &= \frac{m^2}{n}\left(\frac{(P(X=l_g)-1)^2}{P(X=l_g)}+\sum\limits_{\substack{k = 0 \\ k \neq l_g}}^{d}P(X=k)\right) \nonumber\\
     &+ \frac{n-m}{n}\left((1-P(x=l_g)+\sum\limits_{\substack{k = 0 \\ k \neq l_g}}^{d}(1-p(X=k)))\right)\nonumber \\
     &= \frac{m^2}{n}\left(\frac{(P(X=l_g)-1)^2}{P(X=l_g)}+\sum\limits_{\substack{k = 0 \\ k \neq l_g}}^{d}P(X=k)\right) + \frac{(n-m)\cdot d}{n}. \nonumber
\end{align}
Given $\sum\limits_{k=0}^d P(X=k) = 1$, we have $\sum\limits_{\substack{k = 0 \\ k \neq l_g}}^{d}P(X=k) = 1 - P(X=l_g)$. Therefore, 
\begin{align}
    &E_{freq}(O^k_{MGA}, Y^k) \nonumber \\
    &= \frac{m^2}{n}\left(\frac{(P(X=l_g)-1)^2}{P(X=l_g)}+\sum\limits_{\substack{k = 0 \\ k \neq l_g}}^{d}P(X=k)\right) + \frac{(n-m)\cdot d}{n}\nonumber \\
    &= \frac{m^2}{n}\left(\frac{1}{P(X=l_g)}-1\right) + \frac{(n-m)\cdot d}{n}.\nonumber
\end{align}
\end{proof}

\section{Proof of Theorem \ref{errorAPA}} \label{app:proof_errorAPA}
\begin{proof}
For APA, the proof here is similar to that of Theorem~\ref{errorMGA}. We have  
$\mathbb{E}\left[E_{sq}^{APA}(k)\right] = Var(O^k_{APA}) + (\mathbb{E}[O^k_{APA}]-Y^k)^2$.

Let $W^k_{benign}\sim B(n-m, \ P(X=k))$ denote the frequency of bits set to 1 in the reports of $n-m$ benign users and $W^k_{fake}$ be the frequency of bits set to 1 in the reports of $m$ fake users. For APA, $W^k_{fake}=\omega[k]$. Thus, the observed frequency is 
$O^k_{APA} = W^k_{benign} + W^k_{fake}$. Note that $Y^k = nP(X=k)$.

Since $W_{benign}^k$ is independent of $W_{fake}^k$, for each $k\in[0,d]$, we have $$\mathbb{E}\left[O^k_{APA}\right] = \mathbb{E}[W^k_{benign}] + \omega[k] = (n-m)P(X=k) + \omega[k]$$ and $$Var(O^k_{APA}) = Var(W^k_{benign}) = (n-m)P(X=k)(1-P(X=k)).$$
Therefore,
 \begin{align}
     \mathbb{E}\left[E_{sq}^{APA}(k)\right] &= \mathbb{E}\left[(O^k_{APA} - Y^k)^2\right]  \nonumber\\
     &= Var(O^k_{APA}) + (\mathbb{E}[O^k_{APA}]-Y^k)^2\nonumber \\
     &= (n-m)P(X=k)(1-P(X=k)) \nonumber \\
     &+ (m \cdot P(X=k) - \omega[k])^2. \nonumber
 \end{align}
\end{proof}

\section{Proof of Theorem \ref{fiterrorAPA} \label{app:proof_fiterrorAPA}}
\begin{proof}
 \begin{align}
     E_{freq}(O^k_{APA}, Y^k) &= \sum_k\frac{(O^k_{APA} - Y^k)^2}{Y^k}\nonumber\\
     &= \sum\limits_{\substack{k = 0}}^{d} \frac{(O^k_{APA} - Y^k)^2}{Y^k}\nonumber
 \end{align}
 According to Theorem~\ref{errorAPA}, 
 \begin{align}
     &\frac{(O^k_{APA} - Y^k)^2}{Y^k} \nonumber \\
     &= \frac{(m \cdot P(X=k) - \omega[k])^2 + (n-m)P(X=k)(1-P(X=k))}{n\cdot P(X=k)}\nonumber
 \end{align}
As per Eq.~\eqref{eq:fit_error}, we have
 \begin{align}
      &E_{freq}(O^k_{APA}, Y^k) \nonumber \\
      &= \sum\limits_{\substack{k = 0}}^{d} \frac{(m \cdot P(X=k) - \omega[k])^2 + (n-m)P(X=k)(1-P(X=k))}{n\cdot P(X=k)} \nonumber \\
      &= \frac{1}{n}\sum\limits_{\substack{k = 0}}^{d} \frac{(m \cdot P(X=k) - \omega[k])^2}{P(X=k)}+\frac{n-m}{n}\sum\limits_{\substack{k = 0}}^{d}(1-P(X=k))\nonumber \\
      &= \frac{1}{n}\sum\limits_{\substack{k = 0}}^{d} \frac{(m \cdot P(X=k) - \omega[k])^2}{P(X=k)} + \frac{(n-m)\cdot d}{n}. \nonumber
 \end{align}
\end{proof}

\section{Proof of Theorem \ref{errorComp}} \label{app:proof_errorComp}
\begin{proof}

Given Theorems \ref{errorMGA} and \ref{errorAPA}, when $k = l_g$ we have 
\begin{align}
    &\mathbb{E}\left[E^{MGA}_{sq}(k)\right] - \mathbb{E}\left[E^{APA}_{sq}(k)\right] \nonumber \\
    &=  m^2 \cdot(P(X=k) - 1)^2 - (m \cdot P(X=k) - \omega[k])^2   \nonumber \\
    &= m^2\cdot P(X=k)^2 - 2m^2\cdot P(X=k) + m^2 \nonumber \\
    &- (m^2P(X=k)^2-2m\cdot \omega[k] \cdot P(X=k) + \omega[k]^2)\nonumber \nonumber \\
    &=-2m^2\cdot P(X=k) + m^2 + 2m\cdot \omega[k] \cdot P(X=k) - \omega[k]^2 \nonumber \\
    &= (\omega[k] - m) (m \cdot (2P(X=k)-1)-\omega[k]) \nonumber
\end{align}
According to the De Moivre-Laplace theorem \cite{de2020doctrine}, we can approximate $X\sim \mathcal{N}(\mu,\sigma)$ using a normal approximation, where $\mu = dp$ and $\sigma = \sqrt{dp(1-p)}$ with a typical requirement of $dp > 5$. Note that it is common in CFOs to have a large $d$ (e.g., $d>100$). For $\tilde{p}=\frac{p+(d-1)q}{d}$ and $q=\frac{1}{e^\epsilon + 1}\in(0,\frac{1}{2})$, $dp>5$ is easily satisfied when $\tilde{p}\approx q$. Thus, $\mathbb{E}\left[P(X=k=l_g)\right] = \mathbb{E}\left[f(x=\mu)\right] = \frac{1}{\sigma\sqrt{2\pi}}\leq \frac{1}{\sqrt{2\pi}} \approx 0.4$. For $\omega[k=l_g] \in [0, m]$, $R = \mathbb{E}\left[E^{MGA}_{sq}(k)\right] - \mathbb{E}\left[E^{APA}_{sq}(k)\right] $ is 0 when $\omega[k=l_g] = m$; $R>0$ when $\omega[k=l_g] < m$. %\sun{do we need equation numbers 5 and 6?}
\end{proof}

\section{Proof of Theorem \ref{threshold_ASD}}\label{app:proof_threshold_ASD}
\begin{proof}
    Since $\tilde{C}_i \sim N(n\cdot f_i, \ \sigma_i^2)$, we can obtain
, $\xi(\gamma) = \mu_i + Z_\gamma \cdot \sigma_i$ for the item $i$ and a given confidence level $\gamma$, where $\mu_i = n \cdot f_i$ and $Z_\gamma = \Phi^{-1}(\frac{1+\gamma}{2})$ is the $z$-score corresponding to the confidence level $\gamma$ under the standard normal distribution. This upper bound ensures that $\tilde{C}_i$ corresponding to its true frequency $\mu_i$ lies within the interval with the specified confidence level. For items with $f_i = 0$, $\tilde{C}_i$ follows $N(0,\sigma_i^2)$ and $\sigma_i^2 = \frac{n(q(1-q))}{(p-q)^2}$. As a result, the upper bound of $\mathcal{B}$ is $\xi(\gamma) = Z(\gamma) \cdot \sqrt{\frac{nq(1-q)}{(p-q)^2}}$ with confidence $\gamma$.
\end{proof}

\section{Error Analysis of \ASD}\label{app:ASD_error}

%\begin{theorem}
% For a threshold $\xi(\gamma)$, the error $Err$ between $\sum_{\tilde{C}_i \in \mathcal{A}} \tilde{C}_i$ and $n$ is determined by the expectation of the sum of the genuine frequencies $f_i$, given by
%    \begin{align}
%         Err &\leq \sum_{i=1}^{d}f_i \cdot \Pr(\tilde{C}_i < \xi(\gamma)|f_i)
         % &= \sum_{i=1}^{\xi(\gamma)}f_i \cdot \int_{-\infty}^{\xi(\gamma)} \frac{1}{\sqrt{2\pi \frac{\sigma_i^2}{(p-q)^2}}}\exp\left(-\frac{(\tilde{C}_j-n\cdot f_i)^2}{2\frac{\sigma_i^2}{(p-q)^2}}\right)
%    \end{align}
%    where $\Pr(\tilde{C}_i < \xi(\gamma)|f_i)$ denotes the probability that, given a true frequency $f_i$, the corresponding $\tilde{C}_i$ falls within the range $(-\infty, \xi(\gamma)]$.
%\end{theorem}
% The error $Err$ represents the deviation between the sum of $\tilde{C}_i$ exceeding the threshold $\xi(\gamma)$ and the total number of users n.}

    We analyze the error $Err$ in \ASD below. Since $\tilde{C}_i \sim N(n\cdot f_i, \ \sigma_i^2)$, for a given $\tilde{C}_j$, its probability of coming from different  frequencies $f_i$ is 
\begin{equation*}
    f(\tilde{C}_j|f_i) = \frac{1}{\sqrt{2\pi \frac{\sigma_i^2}{(p-q)^2}}}\exp\left(-\frac{(\tilde{C}_j-n\cdot f_i)^2}{2\frac{\sigma_i^2}{(p-q)^2}}\right)
\end{equation*}
Therefore, we can derive the error $Err$
\begin{align}\label{true_error}
    Err &= \mathbb{E} \left[n - n\cdot \sum_{i=1}^{d}\tilde{C}_i \cdot \Pr(\tilde{C}_i>\xi(\gamma)|f_i)\right]\nonumber\\ 
    &= n - \sum_{i=1}^{d}\left[n\cdot f_i \cdot (1-\Pr(\tilde{C}_i>\xi(\gamma)|f_i)) + \sigma_i \phi\left(\frac{\xi(\gamma)-n\cdot f_i}{\sigma_i}\right)\right]\nonumber\\ 
    &= n\cdot \sum_{i=1}^{d}f_i \cdot \Pr(\tilde{C}_i<\xi(\gamma)|f_i)- \sum_{i=1}^{d} \sigma_i \phi\left(\frac{\xi(\gamma)-n\cdot f_i}{\sigma_i}\right) \nonumber  \\  
    &= n\cdot \sum_{i=1}^{d}f_i \cdot \int_{-\infty}^{\xi(\gamma)} f(\tilde{C}_i|f_i) \ d\tilde{C}_j - \sum_{i=1}^{d} \sigma_i \phi\left(\frac{\xi(\gamma)-n\cdot f_i}{\sigma_i}\right) \nonumber \\
    &= n\cdot \sum_{i=1}^{d}f_i \cdot \int_{-\infty}^{\xi(\gamma)} \frac{1}{\sqrt{2\pi \frac{\sigma_i^2}{(p-q)^2}}}\exp\left(-\frac{(\tilde{C}_i-n\cdot f_i)^2}{2\frac{\sigma_i^2}{(p-q)^2}} \right)\nonumber\\
    &- \sum_{i=1}^{d} \sigma_i \phi\left(\frac{\xi(\gamma)-n\cdot f_i}{\sigma_i}\right) 
    % &\leq n\cdot \sum_{i=1}^{d}f_i \cdot \int_{-\infty}^{\xi(\gamma)} \frac{1}{\sqrt{2\pi \frac{\sigma_i^2}{(p-q)^2}}}\exp\left(-\frac{(\tilde{C}_i-n\cdot f_i)^2}{2\frac{\sigma_i^2}{(p-q)^2}} \right)\nonumber \\
    % &= n\cdot \sum_{i=1}^{d}f_i \cdot \Pr(\tilde{C}_i < \xi(\gamma)|f_i)
\end{align}
where $\phi(z)$ is the probability density function of the standard normal distribution. As the detector lacks knowledge of the genuine frequency $f_i$, it is challenging to compute the accurate $Err$.

\section{Additional Results of Recovery}\label{app:recovery}
%Figure~\ref{Recovery_noattack_MSE_emoji&fire} illustrates additional results for the \MSE of different post-processing methods without attacks with emoji and fire datasets. The results are consistent with those observed with zipf. Our proposed \RSN performs the best by showing the lowest \MSE in most settings.

Figures~\ref{Recovery_APA_IGR&MSE_OLHU} and~\ref{Recovery_APA_IGR&MSE_HSTU} present additional recovery results for OLH-User and HST-User. LDPRecover performs similarly to Norm-Sub and \RSN with the zipf dataset but falls behind \RSN with emoji and fire. \RSN consistently demonstrates better \MSE performance, particularly with fire.

%\begin{figure*}[tbp]
%\centering
%\includegraphics[scale=0.33]{figure/MSE_Recovery_unattack_Emoji&Fire.pdf}
% \caption{Utility boost of different post-processing methods without attack in emoji and fire datasets.}
% \label{Recovery_noattack_MSE_emoji&fire}
 %\Description{Utility boost of different post-processing methods without attack in emoji and fire datasets.}
%\end{figure*}

\begin{figure*}[htbp]
\centering
\includegraphics[scale=0.33]{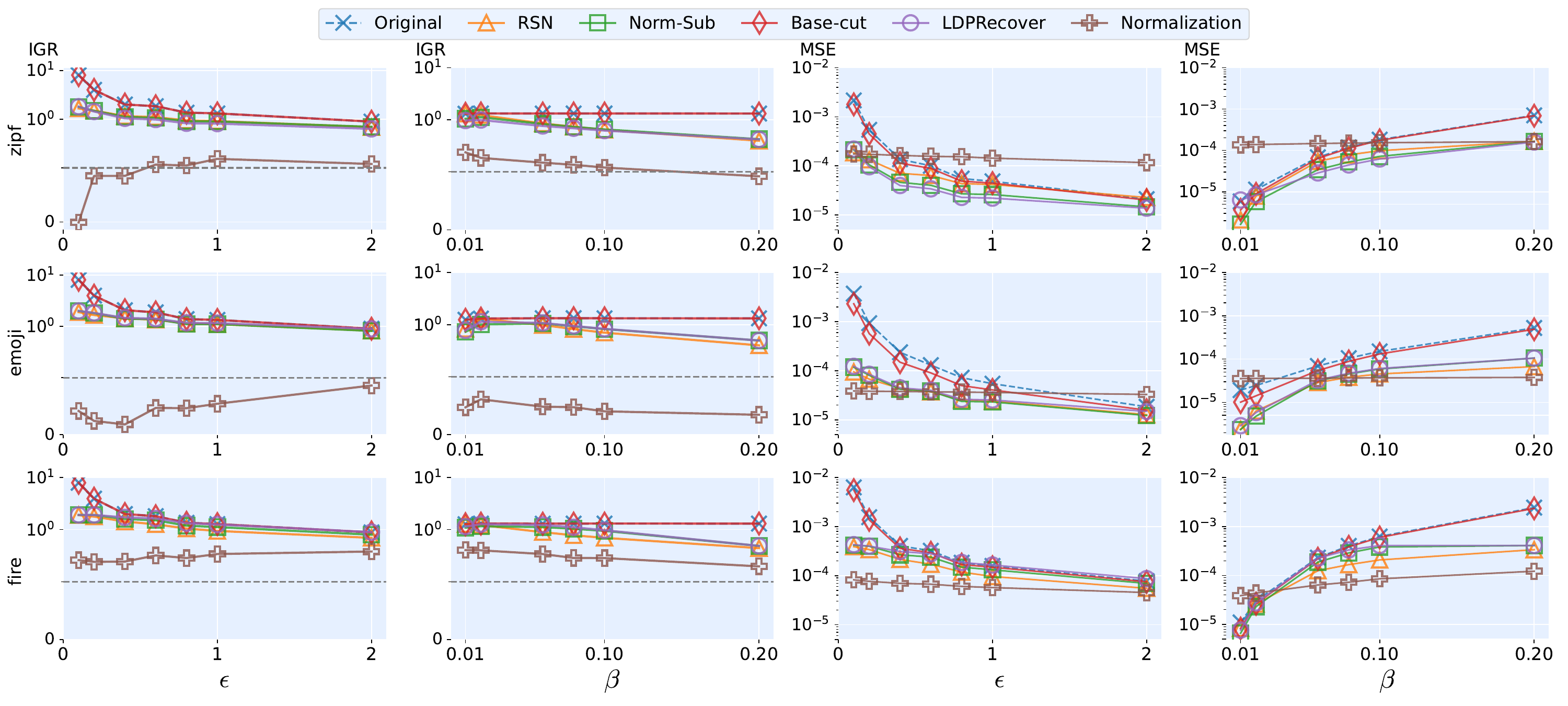}
 \caption{Utility recoverability of different post-processing methods under APA attack on OLH-User.}
 \label{Recovery_APA_IGR&MSE_OLHU}
 \Description{Recoverability of different post-processing methods under APA attack on OLH-User.}
\end{figure*}

\begin{figure*}[htbp]
\centering
\includegraphics[scale=0.35]{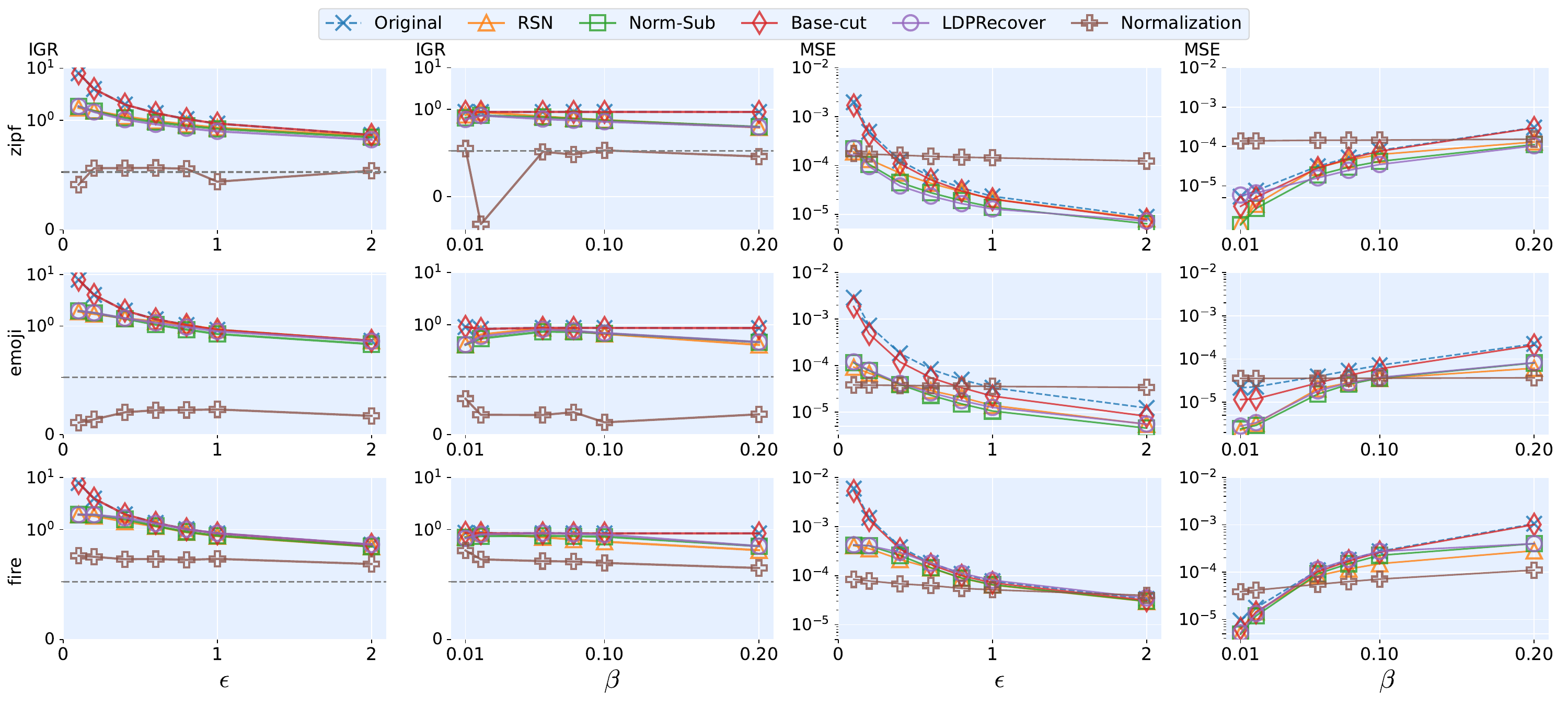}
 \caption{Utility recoverability of different post-processing methods under APA attack on HST-User.}
 \label{Recovery_APA_IGR&MSE_HSTU}
 \Description{Recoverability of different post-processing methods under APA attack on HST-User.}
\end{figure*}
 % \begin{figure*}[h]
 % \centering
 % \includegraphics[scale=0.35]{figure/MSE_Recovery.pdf}
 %  \caption{MSE results of different Post processing methods under APA with different parameter variations. During $\epsilon$ variation, $\beta$ is fixed to 0.05; during $\beta$ variation, $\epsilon$ is fixed to 1.}
 %  \label{Recovery_APA_MSE}
 % \end{figure*}

 %  \begin{figure*}[h]
 % \centering
 % \includegraphics[scale=0.35]{figure/IGR_Recovery.pdf}
 %  \caption{IGR results of different Post processing methods under APA with different parameter variations. During $\epsilon$ variation, $\beta$ is fixed to 0.05; during $\beta$ variation, $\epsilon$ is fixed to 1.}
 %  \label{Recovery_APA_IGR}
 % \end{figure*}
Figures~\ref{Recovery_GRR_IGR} and~\ref{Recovery_GRR_MSE} present the complete recovery results for MGA and MGA-A attacks under the GRR protocol with all three datasets. The results are in line with the conclusion in Section~\ref{sec:Results_ASD}.

\begin{figure*}[htbp]
 \centering
 \includegraphics[scale=0.35]{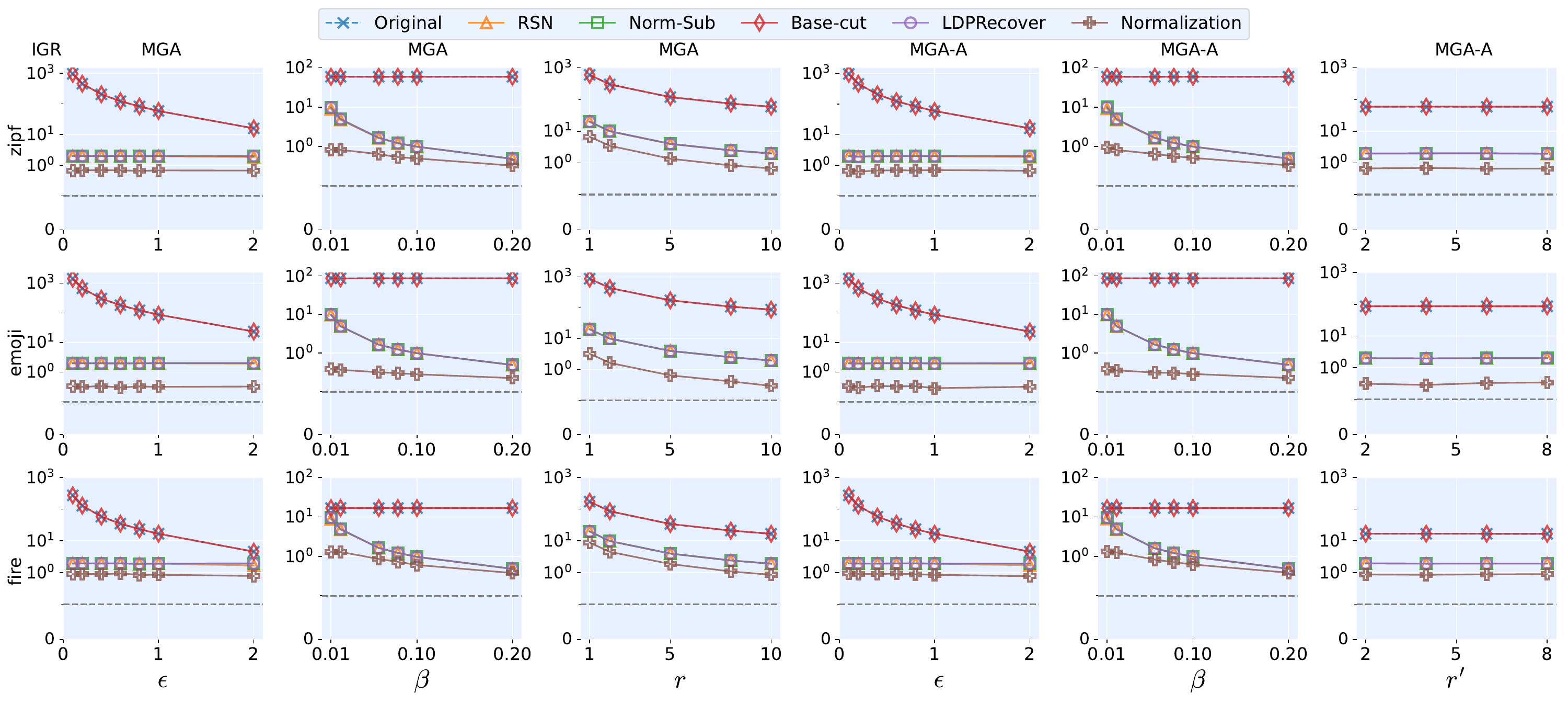}
  \caption{IGR for different post-processing methods under MGA and MGA-A attacks on GRR with emoji and fire. }%\sun{lower case dataset name}}
  \label{Recovery_GRR_IGR}
  \Description{IGR for different post-processing methods under MGA and MGA-A attacks on GRR with emoji and fire.} %\sun{lower case dataset name}}
 \end{figure*}

 \begin{figure*}[htbp]
 \centering
 \includegraphics[scale=0.35]{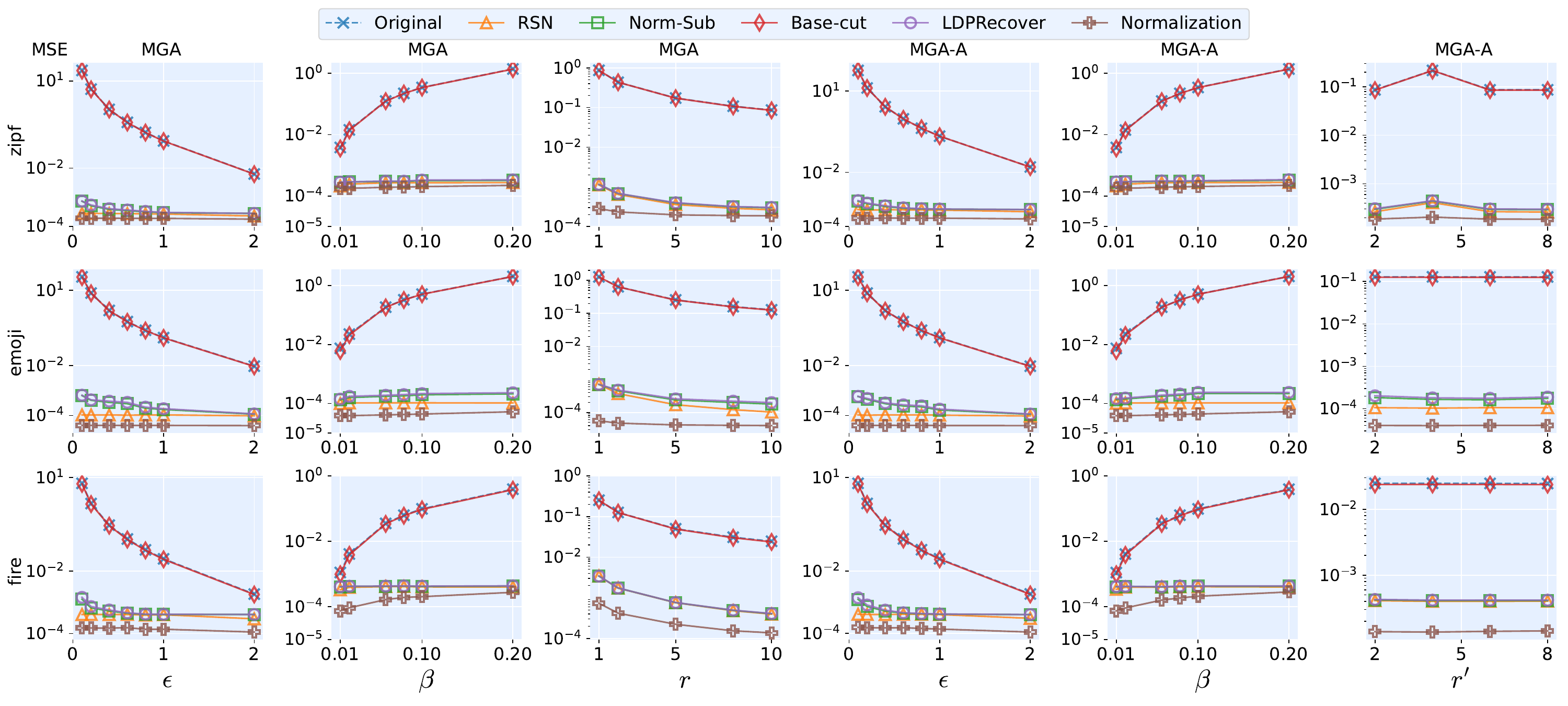}
  \caption{MSE for different post-processing methods under MGA and MGA-A attacks on GRR with emoji and fire. }%\sun{lower case dataset}}
  \label{Recovery_GRR_MSE}
  \Description{MSE for different post-processing methods under MGA and MGA-A attacks on GRR with emoji and fire.}% \sun{lower case dataset}}
 \end{figure*}

%\sun{remove Figure~\ref{hist_post} compares the reserved frequencies after Norm-Sub, LDPRecover, and \RSN. Among the three, \RSN retains more non-target high frequencies with their values closer to originals, though the MSE is similar to that of Norm-Sub and LDPRecover. }

% The results show that Norm-Sub and LDPRecover tend to set all frequencies of non-target items to 0 under the attack. In contrast, \RSN is more resilient against the attack by retaining more original high-frequency items, though its MSE value is similar to that of Norm-Sub and LDPRecover. 

% In Figure~\ref{hist_post}, we present histograms of the real distribution, aggregated estimation, and post-processed results for different methods. We also highlight the number of non-target items retained after processing. The results show that Norm-Sub and LDPRecover tend to reduce the frequencies of all non-target items to zero under substantial attack impact. In contrast, \RSN demonstrates high robustness, retaining original high-frequency items even under intense attacks, showcasing its superior balance between robustness and utility.

 %\begin{figure*}[htbp]
 %\centering
 %\includegraphics[scale=0.35]{figure/hist_recovery_emoji.pdf}
 %  \caption{Comparison of histograms 
 %  after Norm-Sub, LDPRecover, and RSN. The target item frequencies are highlighted in red. $\epsilon = 0.5, r = 10$ and $\beta = 0.05$.}
 %  \label{hist_post}
   % \Description{Comparison of histograms 
   % after Norm-Sub, LDPRecover, and RSN. The target item frequencies are highlighted in red. $\epsilon = 0.5, r = 10$ and $\beta = 0.05$.}
 %\end{figure*}

\end{appendices}
\end{document}